\documentclass[11pt,letterpaper]{article}
\usepackage{color,ifpdf,latexsym,graphicx,url,amsthm,amsmath,amssymb,subcaption,float}

\usepackage{array}
\newcolumntype{C}[1]{>{\centering\let\newline\\\arraybackslash\hspace{0pt}}m{#1}}

\usepackage{thmtools,thm-restate}

\usepackage{hyperref}

\advance \topmargin by -\headheight
\advance \topmargin by -\headsep
\evensidemargin \oddsidemargin
{\makeatletter \gdef\fps@figure{!htbp}}

\let\realbfseries=\bfseries
\def\bfseries{\realbfseries\boldmath}

\newtheorem{theorem}{Theorem}[section]
\newtheorem{lemma}[theorem]{Lemma}
\newtheorem{corollary}[theorem]{Corollary}
\newtheorem{definition}[theorem]{Definition}
\newtheorem{problem}[theorem]{Problem}

{\makeatletter
 \gdef\xxxmark{%
   \expandafter\ifx\csname @mpargs\endcsname\relax 
     \expandafter\ifx\csname @captype\endcsname\relax 
       \marginpar{xxx}
     \else
       xxx 
     \fi
   \else
     xxx 
   \fi}
 \gdef\xxx{\@ifnextchar[\xxx@lab\xxx@nolab}
 \long\gdef\xxx@lab[#1]#2{\textbf{[\xxxmark #2 ---{\sc #1}]}}
 \long\gdef\xxx@nolab#1{\textbf{[\xxxmark #1]}}
}

\begin{document}

\title{Tree-Residue Vertex-Breaking: a new tool for proving hardness}
\author{
Erik D. Demaine%
    \thanks{MIT Computer Science and Artificial Intelligence Laboratory,
      32 Vassar St., Cambridge, MA 02139, USA,
      \protect\url{edemaine@mit.edu}}
\and
  Mikhail Rudoy%
  \thanks{MIT Computer Science and Artificial Intelligence Laboratory,
      32 Vassar St., Cambridge, MA 02139, USA,
      \protect\url{mrudoy@gmail.com}. Now at Google Inc.}
}
\date{}

\maketitle

\begin{abstract}
In this paper, we introduce a new problem called Tree-Residue Vertex-Breaking (TRVB): given a multigraph $G$ some of whose vertices are marked ``breakable,'' is it possible to convert $G$ into a tree via a sequence of ``vertex-breaking'' operations (replacing a degree-$k$ breakable vertex by $k$ degree-$1$ vertices, disconnecting the $k$ incident edges)? 

We characterize the computational complexity of TRVB with any combination of the following additional constraints: $G$ must be planar, $G$ must be a simple graph, the degree of every breakable vertex must belong to an allowed list $B$, and the degree of every unbreakable vertex must belong to an allowed list $U$.
The two results which we expect to be most generally applicable are that (1) TRVB is polynomially solvable when breakable vertices are restricted to have degree at most $3$; and (2) for any $k \ge 4$, TRVB is NP-complete when the given multigraph is restricted to be planar and to consist entirely of degree-$k$ breakable vertices. To demonstrate the use of TRVB, we give a simple proof of the known result that Hamiltonicity in max-degree-$3$ square grid graphs is NP-hard. 

We also demonstrate a connection between TRVB and the Hypergraph Spanning Tree problem. This connection allows us to show that the Hypergraph Spanning Tree problem in $k$-uniform $2$-regular hypergraphs is NP-complete for any $k \ge 4$, even when the incidence graph of the hypergraph is planar.
\end{abstract}

\section{Introduction}
\label{section:introduction}

In this paper, we introduce the Tree-Residue Vertex-Breaking (TRVB) problem. Given a multigraph $G$ some of whose vertices are marked ``breakable,'' TRVB asks whether it is possible to convert $G$ into a tree via a sequence of applications of the \emph{vertex-breaking} operation: replacing a degree-$k$ breakable vertex with $k$ degree-$1$ vertices, disconnecting the incident edges, as shown in Figure~\ref{figure:breaking_example}.

\begin{figure}[!htb]
    \centering
    \hfill\hfill
    \raisebox{-.5\height}{\includegraphics{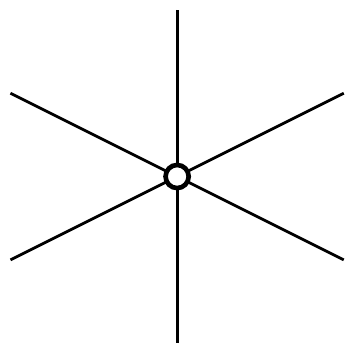}}
    \hfill\raisebox{-.5\height}{\scalebox{2}{$\to$}}\hfill
    \raisebox{-.5\height}{\includegraphics{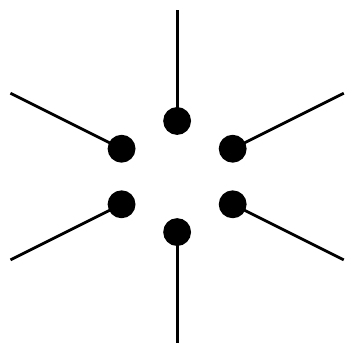}}
    \hfill\hfill\hfill
    \caption{The operation of breaking a vertex. The vertex (left) is replaced by a set of degree-$1$ vertices with the same edges (right).}
    \label{figure:breaking_example}
\end{figure}

In this paper, we analyze the computational complexity of this problem as well as several variants (special cases) where $G$ is restricted with any subset of the following additional constraints:
\begin{enumerate}
\item every breakable vertex of $G$ must have degree from a list $B$ of allowed degrees;
\item every unbreakable vertex of $G$ must have degree from a list $U$ of allowed degrees;
\item $G$ is planar;
\item $G$ is a simple graph (rather than a multigraph).
\end{enumerate} 

Modifying TRVB to include these constraints makes it easier to reduce from the TRVB problem to some other. For example, having a restricted list of possible breakable vertex degrees $B$ allows a reduction to include gadgets only for simulating breakable vertices of those degrees, whereas without that constraint, the reduction would have to support simulation of breakable vertices of any degree. 

We prove the following results (summarized in Table~\ref{table:results_summary}), which together fully classify the variants of TRVB into polynomial-time solvable and NP-complete problems:
\begin{enumerate}
\item Every TRVB variant whose breakable vertices are only allowed to have degrees of at most $3$ is solvable in polynomial time.
\item Every planar simple graph TRVB variant whose breakable vertices are only allowed to have degrees of at least $6$ and whose unbreakable vertices are only allowed to have degrees of at least $5$ is solvable in polynomial time (and in fact the correct output is always ``no'').
\item In all other cases, the TRVB variant is NP-complete. In particular, the TRVB variant is NP-complete if the variant allows breakable vertices of some degree $k \ge 4$, and in the planar graph case, also allows either breakable vertices of some degree $b \le 5$ or unbreakable vertices of some degree $u \le 4$. For example, for any $k \ge 4$, TRVB is NP-complete in planar multigraphs whose vertices are all breakable and have degree $k$.
\end{enumerate}

\begin{table}[]
\def\arraystretch{1.5}
\resizebox{\textwidth}{!}{%
\begin{tabular}{|C{2.8cm}|C{2.5cm}|C{4.0cm}|C{2.6cm}|C{2.5cm}|}
\hline
\bf{All breakable vertices have small degree ($B \subseteq \{1,2,3\}$)} & 
\bf{Graph restrictions} &
\bf{All vertices have large degree ($B \cap \{1,2,3,4\} = \emptyset$ and $U \cap \{1,2,3,4,5\} = \emptyset$)} & 
\bf{TRVB variant complexity} &
\bf{Section} 
\\ \hline
Yes & 
$*$ & 
$*$ & 
Polynomial Time & 
Section~\ref{section:hypergraph}                                                        \\ \hline
No & 
Planar or simple or unrestricted & 
$*$ & 
NP-complete &
Sections \ref{section:hardness_1},~\ref{section:hardness_2},~\ref{section:hardness_3}
\\ \hline
No & 
Planar and simple & 
No & 
NP-complete & 
Section~\ref{section:hardness_4}
\\ \hline
No & 
Planar and simple & 
Yes & 
Polynomial Time (every instance is a ``no'' instance) &
Section~\ref{section:polynomial}                                                        \\ \hline
\end{tabular}
}

\caption{A summary of this paper's results (where $B$ and $U$ are the allowed breakable and unbreakable vertex degrees).}
\label{table:results_summary}
\end{table}

Among these results, we expect the most generally applicable to be the results that (1) TRVB is polynomially solvable when breakable vertices are restricted to have degree at most $3$; and (2) for any $k \ge 4$, TRVB is NP-complete when the given multigraph is restricted to be planar and to consist entirely of degree-$k$ breakable vertices. 

\paragraph{Application to proving hardness.}
In general, the TRVB problem is useful when proving NP-hardness of what could be called \emph{single-traversal problems}: problems in which some space (e.g., a configuration graph or a grid) must be traversed in a single path or cycle subject to local constraints. Hamiltonian Cycle and its variants fall under this category, but so do other problems (e.g., allowing the solution path/cycle to skip certain vertices entirely while still mandating other local constraints). In other words, TRVB can be a useful alternative to Hamiltonian Cycle when proving NP-hardness of problems related to traversal.

To prove a single-traversal problem hard by reducing from TRVB, it is sufficient to demonstrate two gadgets: an edge gadget and a breakable degree-$k$ vertex gadget for some $k \ge 4$. This is because TRVB remains NP-hard even when the only vertices present are degree-$k$ breakable vertices for some $k \ge 4$. Furthermore, since this version of TRVB remains NP-hard even for planar multigraphs, this approach can be used even when the single-traversal problem under consideration involves traversal of a planar space.

One possible approach for building the gadgets is as follows. The edge gadget should contain two parallel paths, both of which must be traversed because of the local constraints of the single-traversal problem (see Figure~\ref{figure:edge_gadget_idea}). The vertex gadget should have exactly two possible solutions satisfying the local constraints of the problem: one solution should disconnect the regions inside all the adjoining edge gadgets, while the other should connect these regions inside the vertex gadget (see Figure~\ref{figure:vertex_gadget_idea}). We then simulate the multigraph from the input TRVB instance by placing these edge and vertex gadgets in the shape of the input multigraph as shown in Figure~\ref{figure:gadgets_arrangement}. 

\begin{figure}[!htb]
    \centering
    \hfill
    \includegraphics{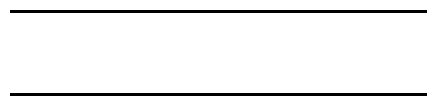}
    \hfill
    \includegraphics{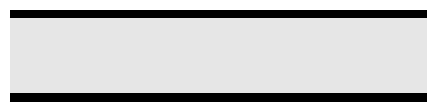}
    \hfill\hfill
    \caption{Abstraction of a possible edge gadget containing two parallel paths (left), together with the local solution of that gadget (right). The bold paths are (forced to be) part of the traversal while the ``inside'' of the gadget is shown in grey.}
    \label{figure:edge_gadget_idea}
\end{figure}

\begin{figure}[!htb]
    \centering
    \hfill
    \includegraphics{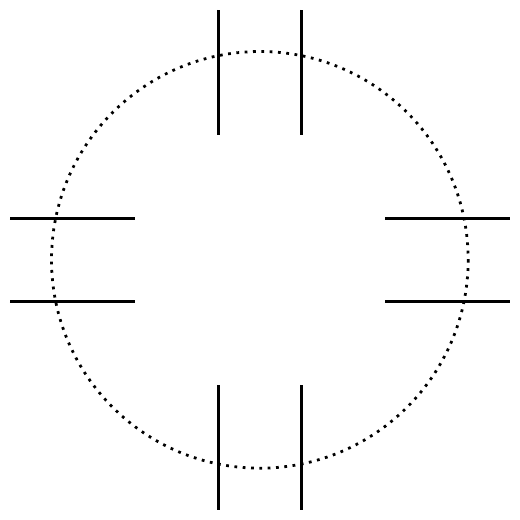}
    \hfill
    \includegraphics{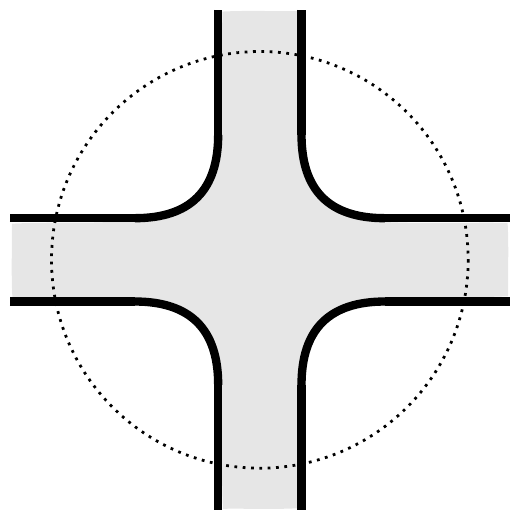}
    \hfill
    \includegraphics{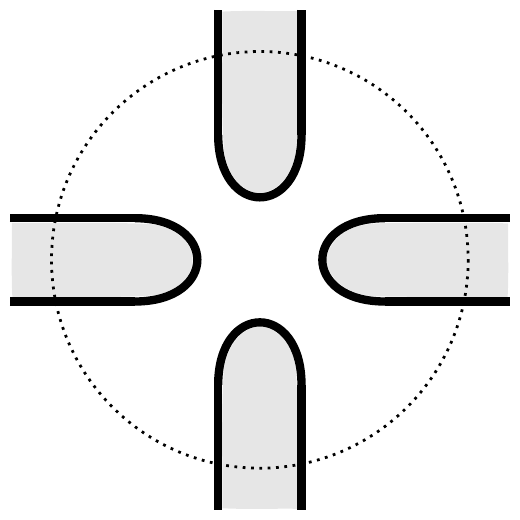}
    \hfill\hfill
    
    \caption{Abstraction of a possible breakable vertex gadget. The gadget should join some number of edge gadgets (in this case four) as shown on the left. The center and right figures show the two possible local solutions to the breakable vertex gadget. One solution connects the interiors of the incoming edge gadgets within the vertex gadget while the other disconnects them. In both figures, the bold paths are part of the traversal, while the ``inside'' of the gadget is shown in grey.}
    \label{figure:vertex_gadget_idea}
\end{figure}

\begin{figure}[!htb]
    \centering
    \hfill\hfill
    \raisebox{-.5\height}{\includegraphics[scale=.5]{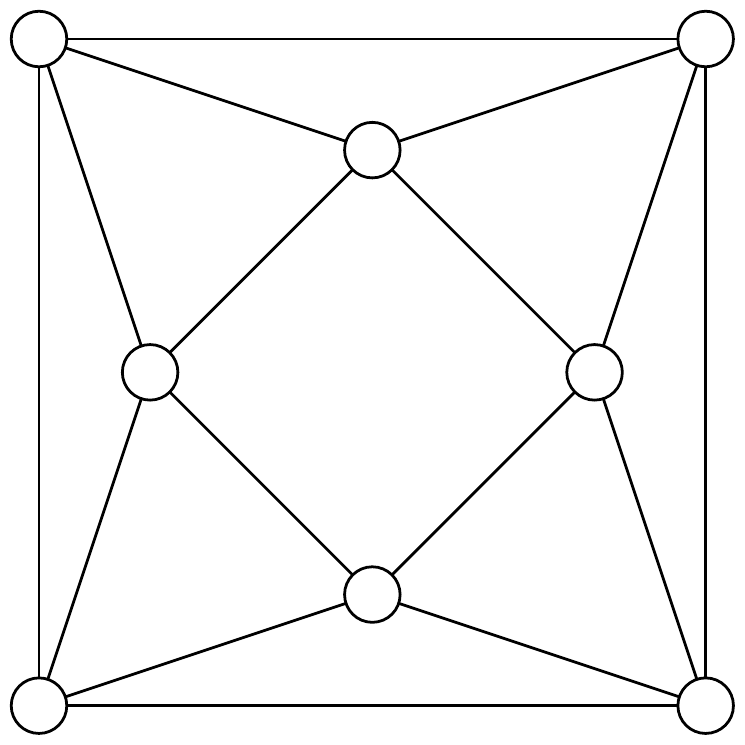}}
    \hfill\raisebox{-.5\height}{\scalebox{2}{$\to$}}\hfill
    \raisebox{-.5\height}{\includegraphics[scale=.5]{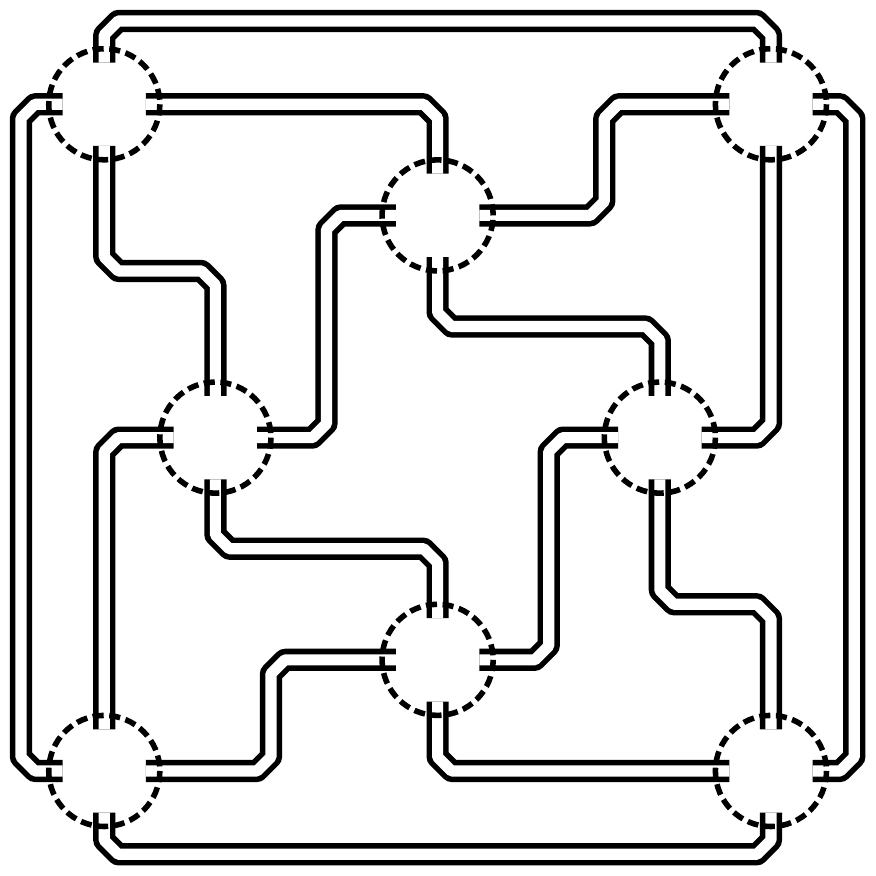}}
    \hfill\hfill\hfill
    \caption{The input multigraph on the left could be converted into a layout of edge and vertex gadgets as shown on the right. In this example, we use a grid layout; in general, we could use any layout consistent with the edge and vertex gadgets.}
    \label{figure:gadgets_arrangement}
\end{figure}

When trying to solve the resulting single-traversal instance, the only option (while satisfying local constraints) is to choose one of the two possible local solutions at each vertex gadget, corresponding to the choice of whether to break the vertex. The candidate solution produced will satisfy all local constraints, but might still not satisfy the global (single cycle) constraint. Notice that the candidate solution is the boundary of the region ``inside'' the local solutions to the edge and vertex gadgets, and that this region ends up being the same shape as the multigraph obtained after breaking vertices. See Figure~\ref{figure:solution_correspondence} for an example. The boundary of this region is a single cycle if and only if the region is connected and hole-free. Since the shape of this region is the same as the shape of the multigraph obtained after breaking vertices, this condition on the region's shape is equivalent to the condition that the residual multigraph must be connected and acyclic, or in other words, a tree. Thus, this construction yields a correct reduction, and in general this proof idea can be used to show NP-hardness of single-traversal problems.

\begin{figure}[!htb]
    \centering
    \hfill\hfill
    \raisebox{-.5\height}{\includegraphics[scale=.5]{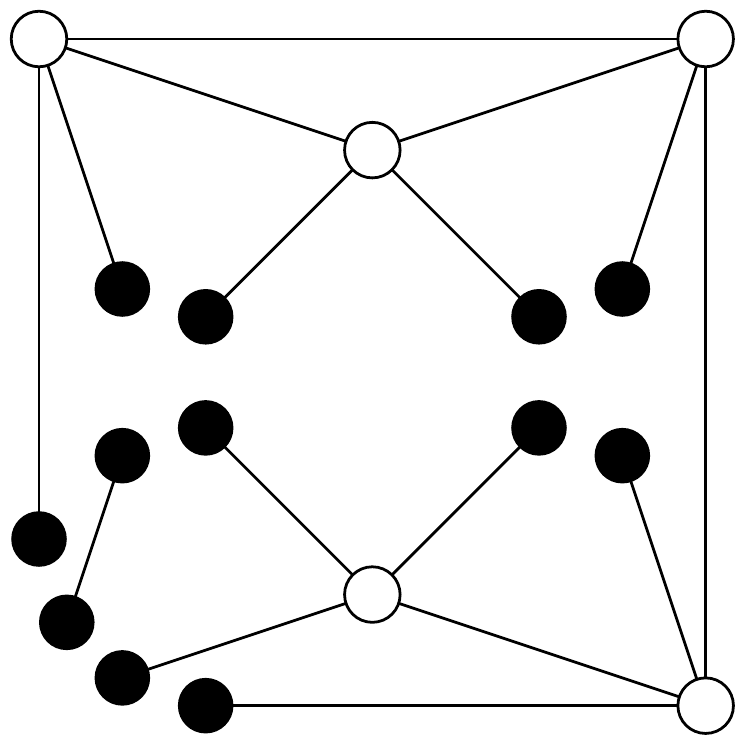}}
    \hfill\raisebox{-.5\height}{\scalebox{2}{$\to$}}\hfill
    \raisebox{-.5\height}{\includegraphics[scale=.5]{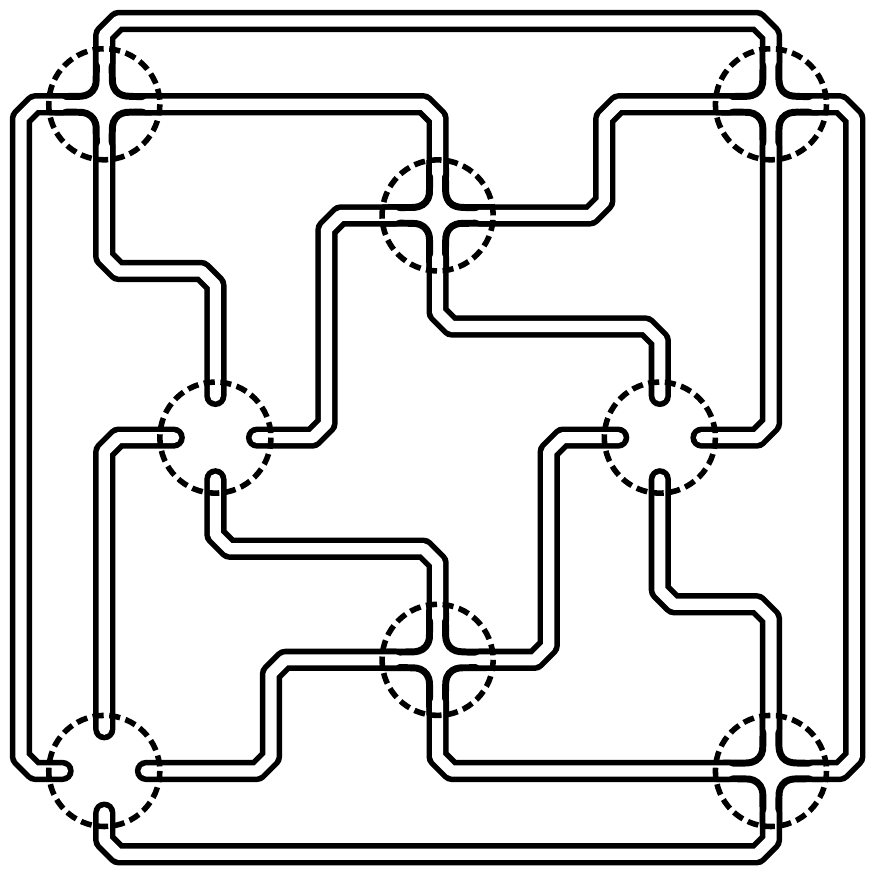}}
    \hfill\raisebox{-.5\height}{\scalebox{2}{$\to$}}\hfill
    \raisebox{-.5\height}{\includegraphics[scale=.5]{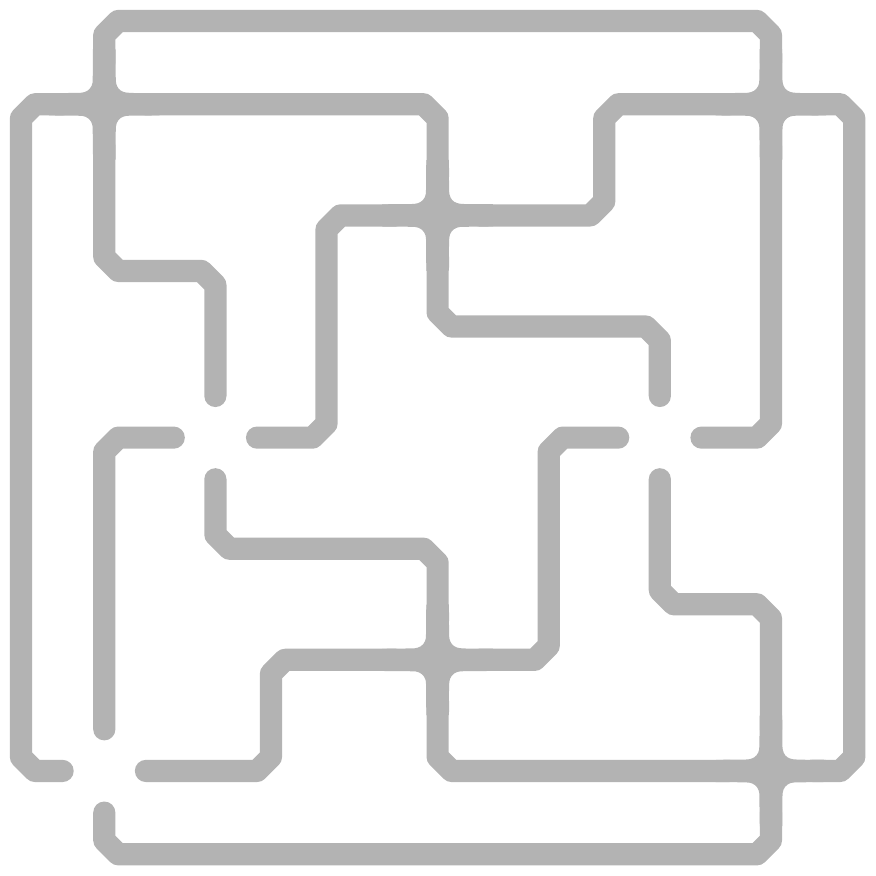}}
    \hfill\hfill\hfill
    \caption{A choice of which vertices to break in the input multigraph (left) corresponds to a choice of local solutions at each of the breakable vertex gadgets, thereby yielding a candidate solution to the single-traversal instance (center). As a result, the shape of the interior of the candidate solution (right) is essentially the same as the shape of the residual multigraph after breaking vertices.}
    \label{figure:solution_correspondence}
\end{figure}

\paragraph{Outline.}
In Section~\ref{section:how_to_use}, we give an example of an NP-hardness proof following the above strategy. By reducing from TRVB, we give a simple proof that Hamiltonian Cycle in max-degree-$3$ square grid graphs is NP-hard (a result previously shown in \cite{Papadimitriou-Vazirani-1984}). We also use the same proof idea in manuscript \cite{GridGraph} to show the novel result that Hamiltonian Cycle in hexagonal thin grid graphs is NP-hard.

In Section~\ref{section:prelim}, we formally define the variants of TRVB under consideration. We also prove membership in NP for all the variants and provide the obvious reductions between them. 

Sections~\ref{section:hardness_1}--\ref{section:hardness_4} address our NP-hardness results. In Section~\ref{section:hardness_1}, we reduce from an NP-hard problem to show that Planar TRVB with only degree-$k$ breakable vertices and unbreakable degree-$4$ vertices is NP-hard for any $k \ge 4$. All the other hardness results in this paper are derived directly or indirectly from this one. In Section~\ref{section:hardness_2}, we prove the NP-completeness of the variants of TRVB and of Planar TRVB in which breakable vertices of some degree $k \ge 4$ are allowed. Similarly, we show in Section~\ref{section:hardness_3} that Graph TRVB is also NP-complete in the presence of breakable vertices of degree $k \ge 4$. Finally, in Section~\ref{section:hardness_4}, we show that Planar Graph TRVB is NP-complete provided (1) breakable vertices of some degree $k \ge 4$ are allowed and (2) either breakable vertices of degree $b \le 5$ or unbreakable vertices of degree $u \le 4$ are allowed.


Next, in Section~\ref{section:polynomial}, we proceed to one of our polynomial-time results: that a variant of TRVB is solvable in polynomial time whenever the multigraph is restricted to be a planar graph, the breakable vertices are restricted to have degree at least $6$, and the unbreakable vertices are restricted to have degree at least $5$. In such a graph, it is impossible to break a set of breakable vertices and get a tree. As a result, variants of TRVB satisfying these restrictions are always solvable with a trivial polynomial time algorithm.

In Section~\ref{section:hypergraph}, we establish a connection between TRVB and the Hypergraph Spanning Tree problem (given a hypergraph, decide whether it has a spanning tree). Namely, Hypergraph Spanning Tree on a hypergraph is equivalent to TRVB on the corresponding incidence graph with edge nodes marked breakable and vertex nodes marked unbreakable. This equivalence allows us to construct a reduction from TRVB to Hypergraph Spanning Tree: given a TRVB instance, we can first convert that instance into a bipartite TRVB instance (by inserting unbreakable vertices between adjacent breakable vertices and merging adjacent unbreakable vertices) and then construct the hypergraph whose incidence graph is the bipartite TRVB instance. 


This connection allows us to obtain results about both TRVB and Hypergraph Spanning Tree. By leveraging known results about Hypergraph Spanning Tree (see \cite{lovasz}), we prove that TRVB is polynomial-time solvable when all breakable vertices have small degrees ($B \subseteq \{1, 2, 3\}$). This final result completes our classification of the variants of TRVB. We also apply the hardness results from this paper to obtain new results about Hypergraph Spanning Tree; namely, Hypergraph Spanning Tree is NP-complete in $k$-uniform $2$-regular hypergraphs for any $k \ge 4$, even when the incidence graph of the hypergraph is planar. This improves the previously known result that Hypergraph Spanning Tree is NP-complete in $k$-uniform hypergraphs for any $k \ge 4$ (see \cite{promel2002steiner}).


\section{Example of how to use TRVB: Hamiltonicity in max-degree-$3$ square grid graphs}
\label{section:how_to_use}

In this section, we show one example of using TRVB to prove hardness of a single-traversal problem. Namely, the result that Hamiltonian Cycle in max-degree-$3$ square grid graphs is NP-hard \cite{Papadimitriou-Vazirani-1984} can be reproduced with the following much simpler reduction.

The reduction is from the variant of TRVB in which the input multigraph is restricted to be planar and to have only degree-$4$ breakable vertices, which is shown NP-complete in Section~\ref{section:hardness_2}. Given a planar multigraph $G$ with only degree-$4$ breakable vertices, we output a max-degree-$3$ square grid graph by appropriately placing breakable degree-$4$ vertex gadgets (shown in Figure~\ref{figure:vertex_gadget_how_to}) and routing edge gadgets (shown in Figure~\ref{figure:edge_gadget_how_to}) to connect them. The appropriate placement of gadgets can be accomplished in polynomial time by the results from \cite{routing}. Each edge gadget consists of two parallel paths of edges a distance of two apart, and as shown in the figure, these paths can turn, allowing the edge to be routed as necessary (without parity constraints). Each breakable degree-$4$ vertex gadget joins four edge gadgets in the configuration shown. Note that, as desired, the maximum degree of any vertex in the resulting grid graph is $3$.

\begin{figure}[!htb]
  \centering
  \begin{minipage}{0.45\linewidth}
    \centering
    \includegraphics[scale=.5]{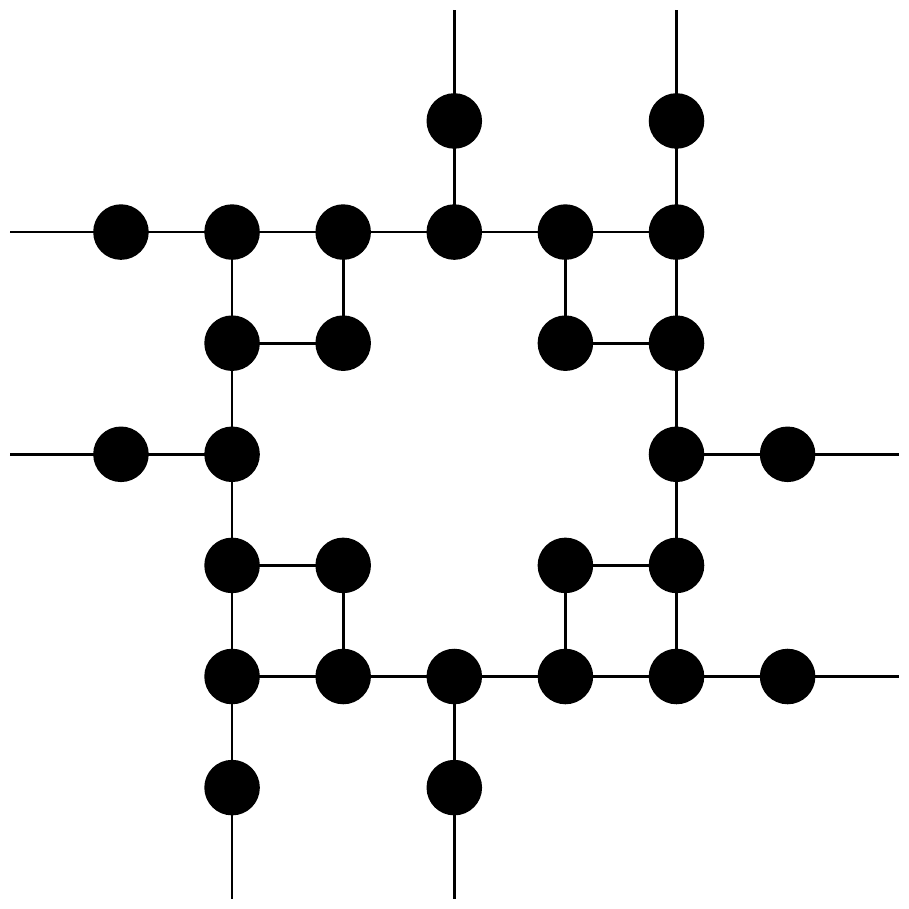}
    \caption{A degree-$4$ breakable vertex gadget.}
    \label{figure:vertex_gadget_how_to}
  \end{minipage}\hfill
  \begin{minipage}{0.45\linewidth}
    \centering
    \includegraphics[scale=.5]{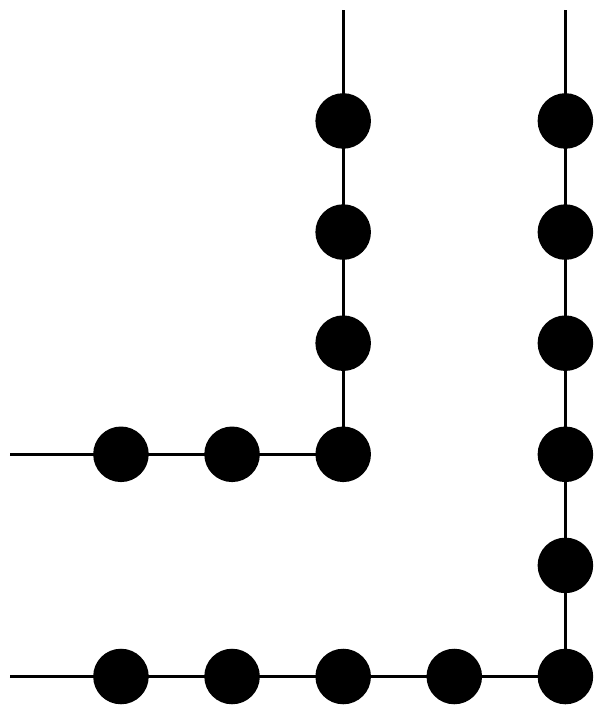}
    \caption{An example edge gadget consisting of two parallel paths of edges a distance of two apart.}
    \label{figure:edge_gadget_how_to}
  \end{minipage}
\end{figure}

Consider any candidate set of edges $C$ that could be a Hamiltonian cycle in the resulting grid graph. In order for $C$ to be a Hamiltonian cycle, $C$ must satisfy both the local constraint that every vertex is incident to exactly two edges in $C$ and the global constraint that $C$ is a cycle (rather than a set of disjoint cycles). It is easy to see that, in order to satisfy the local constraint, every edge in every edge gadget must be in $C$. Similarly, there are only two possibilities within each breakable degree-$4$ vertex gadget which satisfy the local constraint. These possibilities are shown in Figure~\ref{figure:vertex_gadget_solution_how_to}.

\begin{figure}[!htb]
    \centering
    \includegraphics[scale=.5]{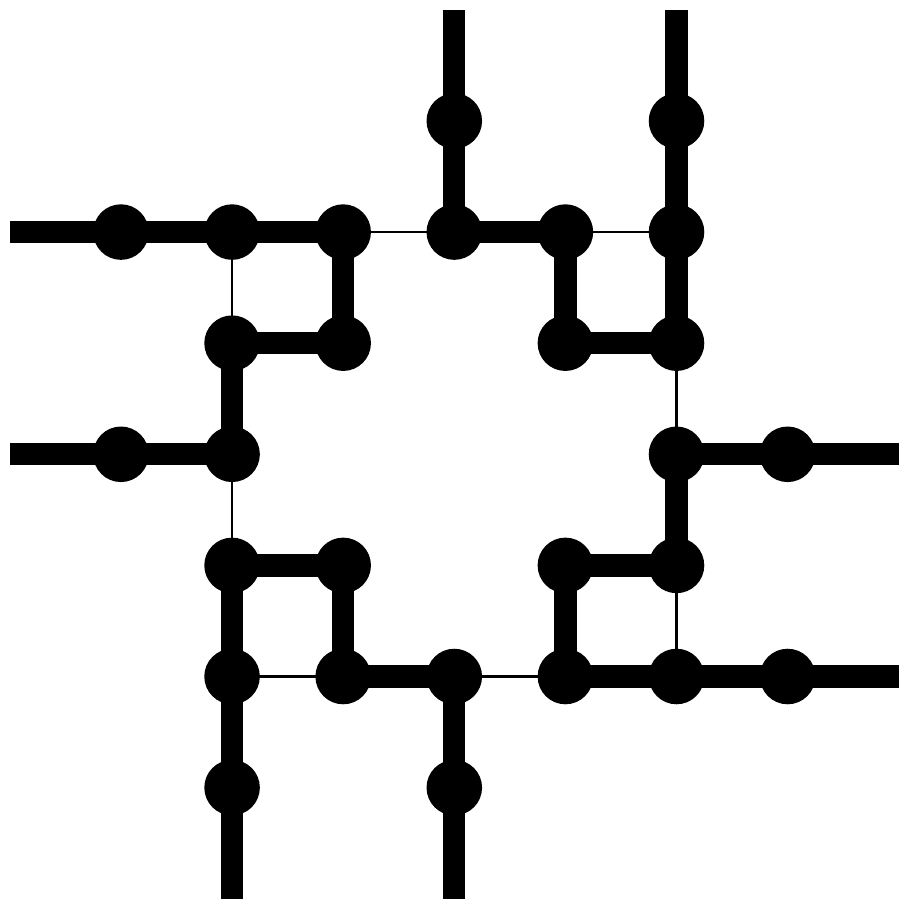}~~
    \includegraphics[scale=.5]{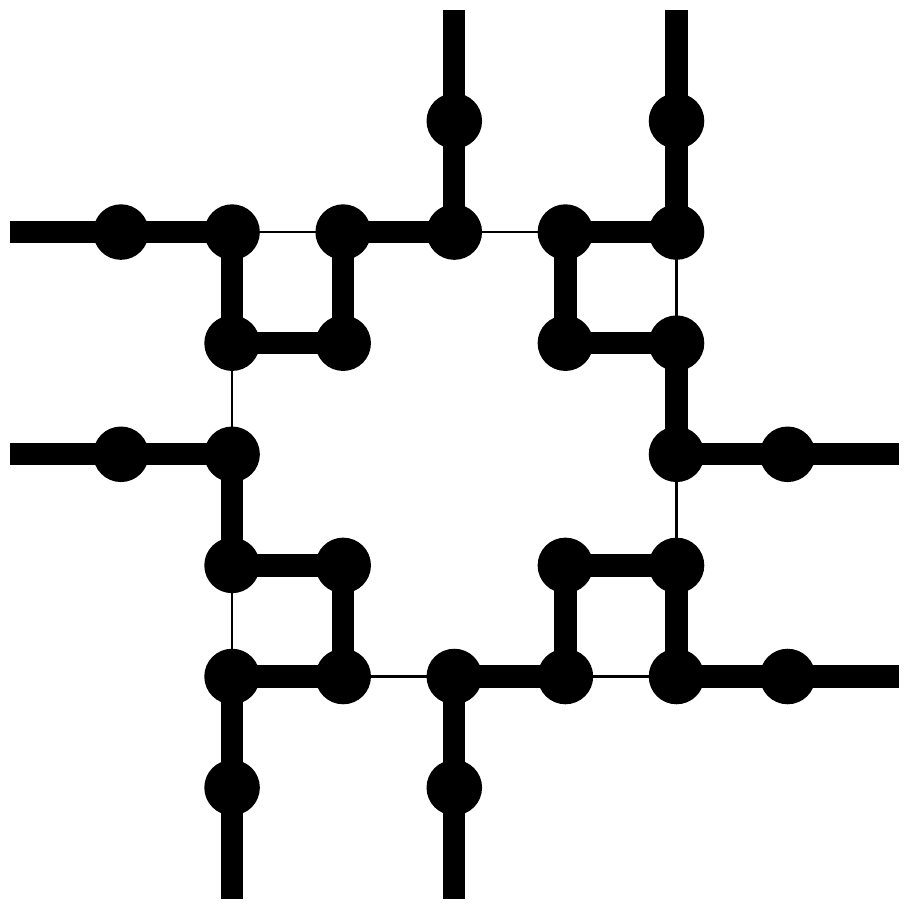}
    \caption{The two possible solutions to the vertex gadget from Figure~\ref{figure:vertex_gadget_how_to} which satisfy the local constraints imposed by the Hamiltonian Cycle problem.}
    \label{figure:vertex_gadget_solution_how_to}
\end{figure}

We can identify the choice of local solution at each breakable degree-$4$ vertex gadget with the choice of whether to break the corresponding vertex. Under this bijection, every candidate solution $C$ satisfying local constraints corresponds with a possible multigraph $G'$ formed from $G$ by breaking vertices. The key insight is that the shape of the region $R$ inside $C$ is exactly the shape of $G'$. This is shown for an example graph-piece in Figure~\ref{figure:graph_piece_example_how_to}. The boundary of $R$, also known as $C$, is exactly one cycle if and only if $R$ is connected and hole-free. Since the shape of region $R$ is the same as the shape of multigraph $G'$, this corresponds to the condition that $G'$ is connected and acyclic, or in other words that $G'$ is a tree. Thus, there exists a candidate solution $C$ to the Hamiltonian Cycle instance (satisfying the local constraints) that is an actual solution (also satisfying the global constraints) if and only if $G$ is a ``yes'' instance of TRVB. Therefore, Hamiltonian Cycle in max-degree-$3$ square grid graphs is NP-hard. 

\begin{figure}[!htb]
    \centering
    \hfill
    \includegraphics[width=.2\textwidth]{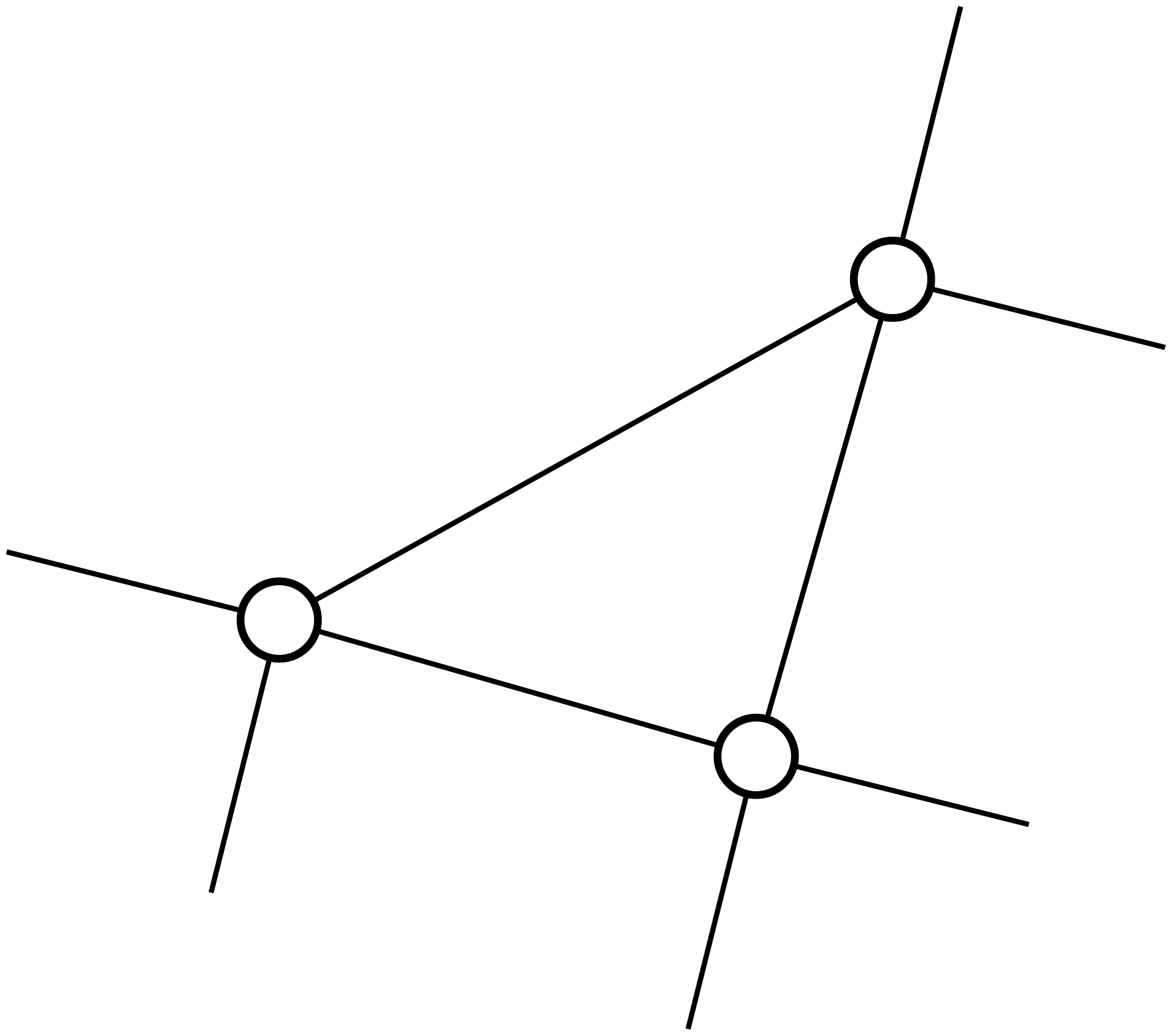}\hfill
    \includegraphics[width=.2\textwidth]{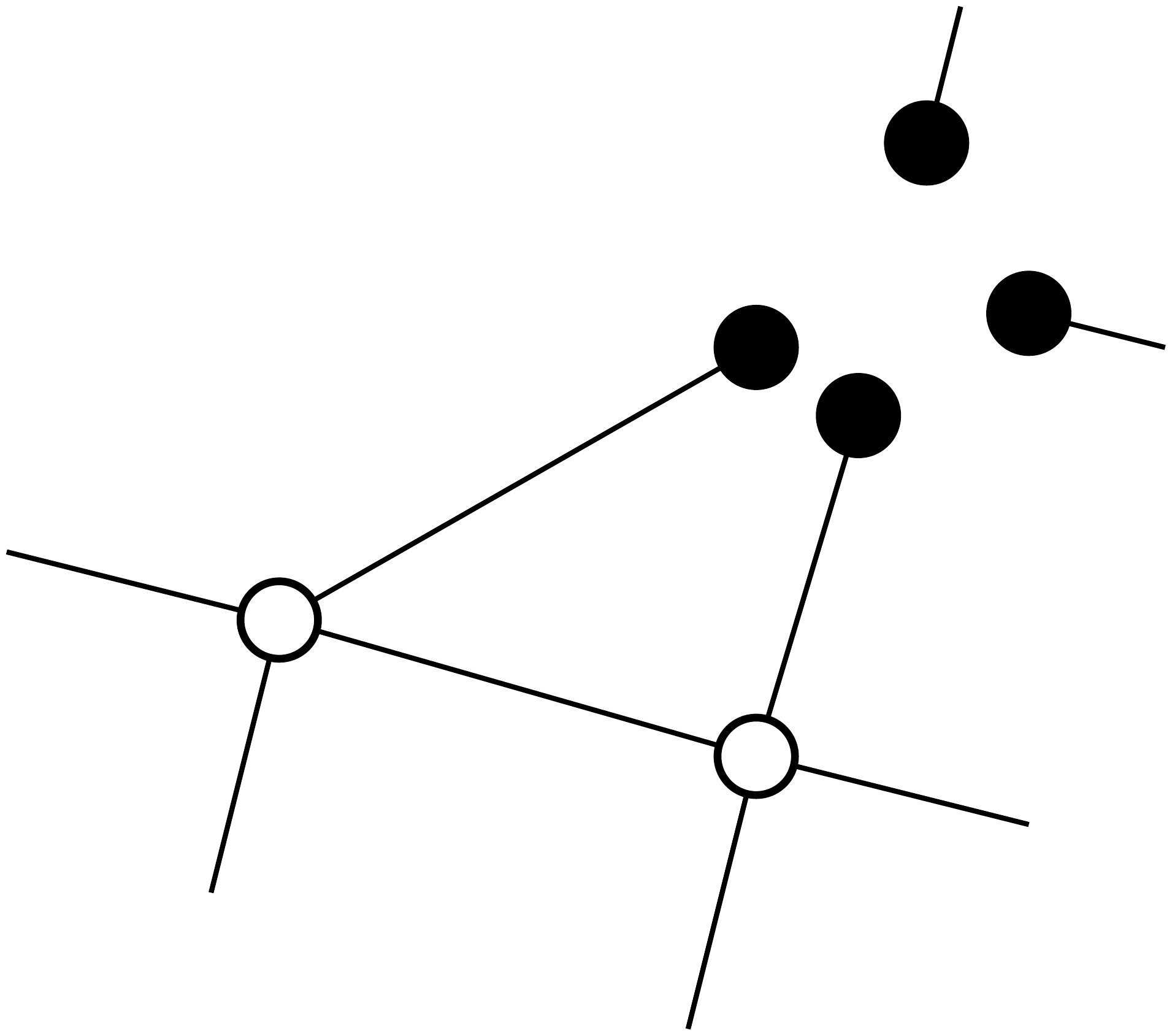}\hfill\hfill

    \hfill
    \includegraphics[width=.45\textwidth]{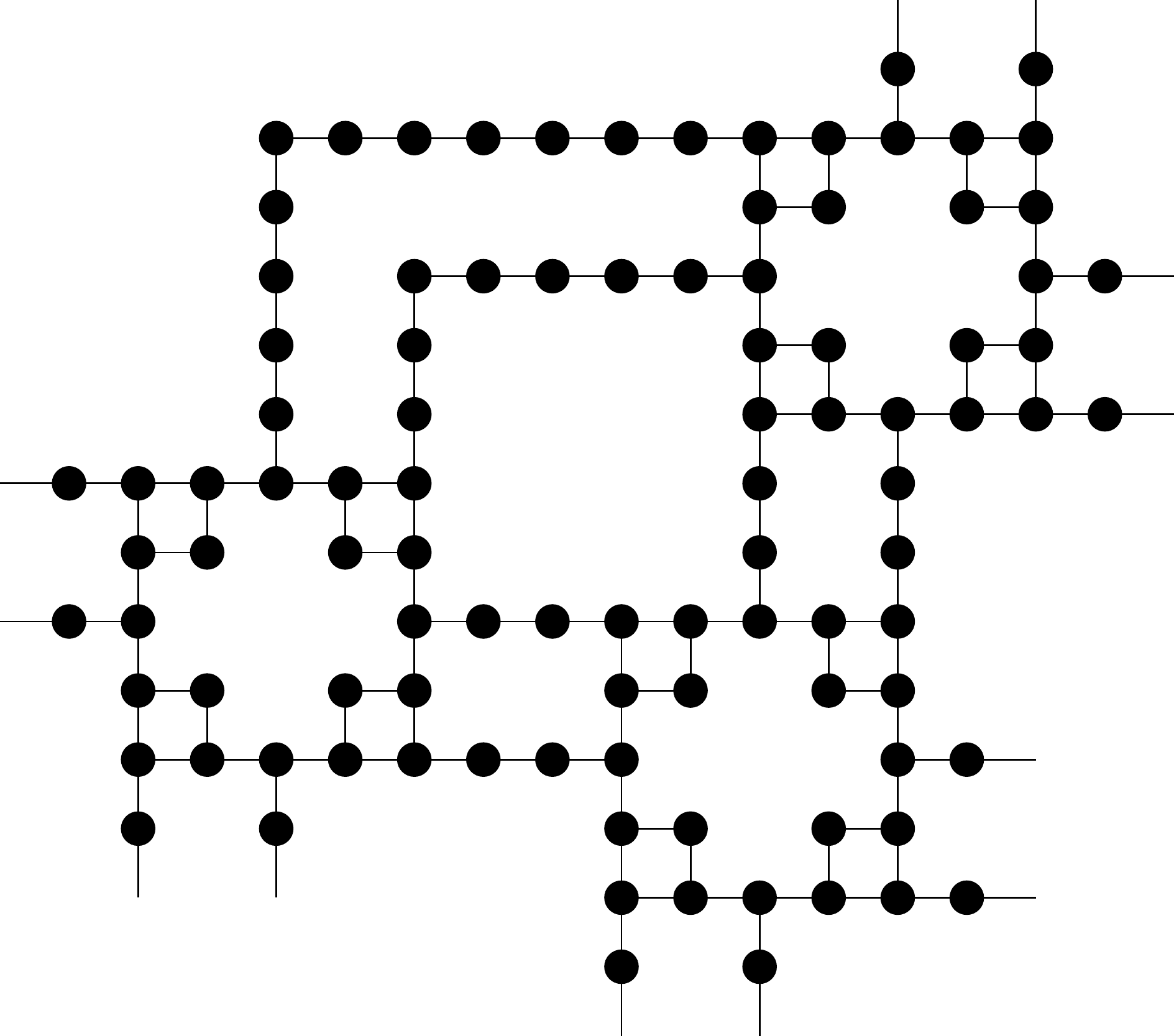}\hfill
    \includegraphics[width=.45\textwidth]{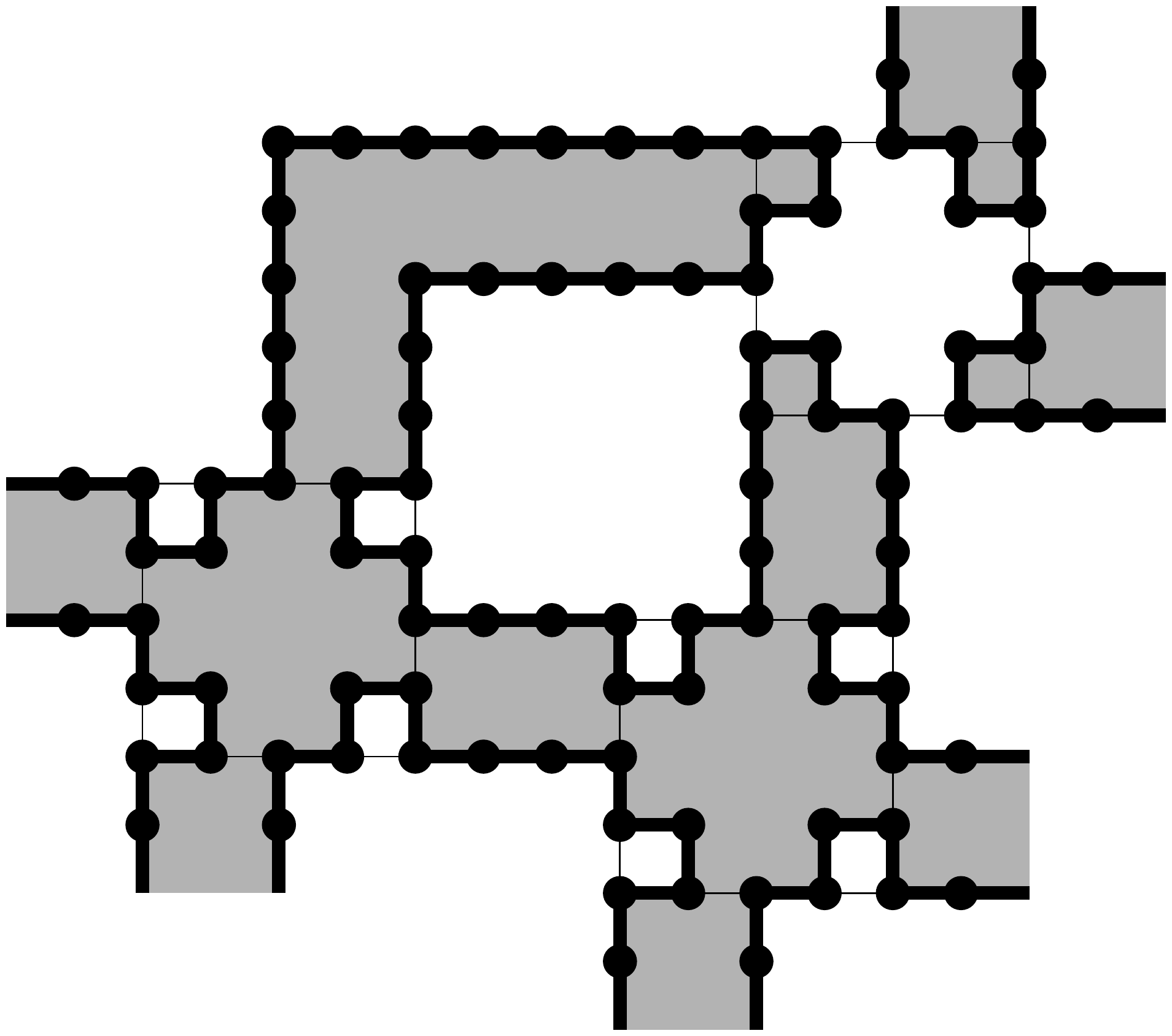}\hfill\hfill
    \caption{Given a multigraph including the piece shown in the top left, the output grid graph might include the section shown in the bottom left (depending on graph layout). If the top vertex in this piece of the multigraph is broken, resulting in the piece of multigraph $G'$ shown in the top right, then the resulting candidate solution $C$ (shown in bold) in the bottom right contains region $R$ (shown in grey) whose shape resembles the shape of $G'$.}
    \label{figure:graph_piece_example_how_to}
\end{figure}

\section{Problem variants}
\label{section:prelim}

In this section, we will formally define the variants of TRVB under consideration, and prove some basic results about them.

To begin, we formally define the TRVB problem. The multigraph operation of \emph{breaking} vertex $v$ in undirected multigraph $G$ results in a new multigraph $G'$ by removing $v$, adding a number of new vertices equal to the degree of $v$ in $G$, and connecting these new vertices to the neighbors of $v$ in $G$ in a one-to-one manner (as shown in Figure~\ref{figure:breaking_example} in Section~\ref{section:introduction}). Using this definition, we pose the TRVB problem:

\begin{problem}
The \emph{Tree-Residue Vertex-Breaking Problem (TRVB)} takes as input a multigraph $G$ whose vertices are partitioned into two sets $V_B$ and $V_U$ (called the \emph{breakable} and \emph{unbreakable} vertices respectively), and asks to decide whether there exists a set $S \subseteq V_B$ such that after breaking every vertex of $S$ in $G$, the resulting multigraph is a tree.
\end{problem}

In order to avoid trivial cases, we consider only input graphs that have no degree-$0$ vertices.

Next, suppose $B$ and $U$ are both sets of positive integers. Then we can constrain the breakable vertices of the input to have degrees in $B$ and constrain the unbreakable vertices of the input to have degrees in $U$. The resulting constrained version of the problem is defined below:

\begin{definition}
The \emph{$(B, U)$-variant of the TRVB problem}, denoted $(B, U)$-TRVB, is the special case of TRVB where the input multigraph is restricted so that every breakable vertex in $G$ has degree in $B$ and every unbreakable vertex in $G$ has degree in $U$. 
\end{definition}

Throughout this paper we consider only sets $B$ and $U$ for which membership can be computed in pseudopolynomial time (i.e., membership of $n$ in $B$ or $U$ can be computed in time polynomial in $n$). As a result, verifying that the vertex degrees of a given multigraph are allowed can be done in polynomial time. This means that the classification of a particular $(B, U)$-variant of the TRVB problem into P or NP-complete is a statement about the hardness of checking TRVB (while constrained by the other conditions) rather than a statement about the hardness of checking membership in $B$ or $U$ for the degrees in the multigraph. In fact, all the results in this paper will also apply even in the cases that membership in $B$ or $U$ cannot be computed in pseudopolynomial time if we consider the promise problems in which the given multigraph's vertex degrees are guaranteed to comply with the sets $B$ and $U$.

We can also define three further variants of the problem depending on whether $G$ is constrained to be planar, a (simple) graph, or both: the \emph{Planar $(B, U)$-variant of the TRVB problem} (denoted Planar $(B, U)$-TRVB), the \emph{Graph $(B, U)$-variant of the TRVB } (denoted Graph $(B, U)$-TRVB), and the \emph{Planar Graph $(B, U)$-variant of the TRVB problem} (denoted Planar Graph $(B, U)$-TRVB).

Note that since both being planar and being a graph are properties of a multigraph that can be verified in polynomial time, again the classification of these variants into P or NP-complete is a statement about the hardness of TRVB.

\subsection{Diagram conventions}

Throughout this paper, when drawing diagrams, we will use filled circles to represent unbreakable vertices and unfilled circles to represent breakable vertices. See Figure~\ref{figure:example_vertices}.

\begin{figure}[!htb]
    \centering
    \includegraphics{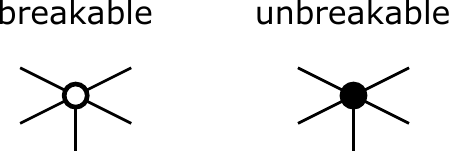}
    \caption{An example diagram, showing the depictions of vertex types used in this paper.}
    \label{figure:example_vertices}
\end{figure}

\subsection{Trivial reductions}

As mentioned above, except for the constraint that the TRVB problem outputs ``yes'' on the given input, every other constraint in the definition of each of the above variants can be tested in polynomial time. Therefore, if for some two variants $X$ and $Y$ the non-TRVB conditions of $X$ are strictly stronger (more constraining) than the non-TRVB conditions of $Y$, then we can reduce from $X$ to $Y$ in polynomial time. In particular, we can convert an input $G$ for variant $X$ into an input $G'$ for $Y$ as follows:

First test all the non-TRVB conditions of variant $X$ on the input $G$. If any condition is not satisfied, then $X$ rejects $G$, so output any $G'$ rejected by $Y$. If all the non-TRVB conditions of variant $X$ are satisfied, then by assumption all the non-TRVB conditions of variant $Y$ on input $G$ are also satisfied. Therefore $G$ is a ``yes'' instance of both $X$ and $Y$ if and only if $G$ is a ``yes'' instance of TRVB. Therefore $X$ and $Y$ have the same answer on $G$, so outputting $G' = G$ completes the reduction.

Using the above reduction scheme, we conclude that:
\begin{lemma}
\label{lemma:trivial_reductions}
For any $(B, U)$, there are reductions
\begin{itemize}
\item from Planar $(B, U)$-TRVB to $(B, U)$-TRVB,
\item from Graph $(B, U)$-TRVB to $(B, U)$-TRVB,
\item from Planar Graph $(B, U)$-TRVB to Planar $(B, U)$-TRVB, and
\item from Planar Graph $(B, U)$-TRVB to Graph $(B, U)$-TRVB.
\end{itemize}
For any $(B, U)$ and $(B', U')$ with $B \subseteq B'$ and $U \subseteq U'$, there are reductions
\begin{itemize}
\item from $(B, U)$-TRVB to $(B', U')$-TRVB,
\item from Planar $(B, U)$-TRVB to Planar $(B', U')$-TRVB,
\item from Graph $(B, U)$-TRVB to Graph $(B', U')$-TRVB, and
\item from Planar Graph $(B, U)$-TRVB to Planar Graph $(B', U')$-TRVB.
\end{itemize}
\end{lemma}

\subsection{Membership in NP}

\begin{theorem}
The TRVB problem is in NP.
\end{theorem}

\begin{proof}
We describe a nondeterministic algorithm to solve TRVB: First nondeterministically guess a set of breakable vertices in $G$. Break that set of vertices and accept if and only if the resulting multigraph is a tree. 

This algorithm accepts an input $G$ on at least one nondeterministic branch if and only if it is possible to break some of the breakable vertices so that the residual multigraph is a tree. In other words, this algorithm solves TRVB. Furthermore, the algorithm runs in polynomial time since both breaking vertices and checking whether a multigraph is a tree are polynomial-time operations. As desired, TRVB is in NP.
\end{proof}

Another name for TRVB is $(\mathbb{N}, \mathbb{N})$-TRVB, so we can apply the reductions from Lemma~\ref{lemma:trivial_reductions} to conclude that:

\begin{corollary}
\label{corollary:np_membership}
For any $(B, U)$, the $(B, U)$-TRVB, Planar $(B, U)$-TRVB, Graph $(B, U)$-TRVB, and Planar Graph $(B, U)$-TRVB problems are in NP.
\end{corollary}

\section{Planar $(\{k\}, \{4\})$-TRVB is NP-hard for any $k \ge 4$}
\label{section:hardness_1}

The overall goal of this section is to prove NP-hardness for several variants of TRVB. In particular, we will introduce an NP-hard variant of the Hamiltonicity problem in Section~\ref{section:hardness_1_a} and then reduce from this problem to Planar $(\{k\}, \{4\})$-TRVB for any $k \ge 4$ in Section~\ref{section:hardness_1_b}. This is the only reduction from an external problem in this paper. All further hardness results will be derived from this one via reductions between different TRVB variants.

\subsection{Planar Hamiltonicity in Directed Graphs with all in- and out-degrees $2$ is NP-hard}
\label{section:hardness_1_a}

The following problem was shown NP-complete in \cite{hamiltonicity}:

\begin{problem}
The \emph{Planar Max-Degree-$3$ Hamiltonicity Problem} asks for a given planar directed graph whose vertices each have total degree at most $3$ whether the graph is Hamiltonian (has a Hamiltonian cycle).
\end{problem} 

For the sake of simplicity we will assume that every vertex in an input instance of the Planar Max-Degree-$3$ Hamiltonicity problem has both in- and out-degree at least $1$ (and therefore at most $2$). This is because the existence of a vertex with in- or out-degree $0$ in a graph immediately implies that there is no Hamiltonian cycle in that graph. 

As it turns out, this problem is not quite what we need for our reduction, so below we introduce several new definitions and define a new variant of the Hamiltonicity problem:

\begin{definition}
Call a vertex $v \in G$ \emph{alternating} for a given planar embedding of a planar directed graph $G$ if, when going around the vertex, the edges switch from inward to outward oriented more than once. Otherwise, call the vertex \emph{non-alternating}. A non-alternating vertex has all its inward oriented edges in one contiguous section and all its outward oriented edges in another; an alternating vertex on the other hand alternates between inward and outward sections more times.

We call a planar embedding of planar directed graph $G$ a \emph{planar non-alternating embedding} if every vertex is non-alternating under that embedding. If $G$ has a planar non-alternating embedding we say that $G$ is a \emph{planar non-alternating graph}.
\end{definition}

\begin{problem}
The \emph{Planar Non-Alternating Indegree-$2$ Outdegree-$2$ Hamiltonicity Problem} asks, for a given planar non-alternating directed graph whose vertices each have in- and out-degree exactly $2$, whether the graph is Hamiltonian
\end{problem} 

The goal of this section is to prove that this problem is NP-hard. To this purpose, consider the following definition and lemmas:

\begin{definition}
Define \emph{simplifying $G$ over edge $(u, v)$} to be the following operation: remove all edges $(u', v)$ and $(u, v')$ from $G$ and then contract edge $(u, v)$. The resulting graph has one new vertex instead of $u$ and $v$; this vertex inherits the inward oriented edges of $u$ and inherits the outward oriented edges of $v$. The inward oriented edges of $v$ and outward oriented edges of $u$ are removed from the graph.
\end{definition}

\begin{lemma}
\label{lemma:simplify_hamiltonicity}
If $(u, v)$ is an edge of directed graph $G$ and either $u$ has outdegree $1$ or $v$ has indegree $1$, then simplifying $G$ over $(u, v)$ maintains the Hamiltonicity of $G$. 
\end{lemma}

\begin{proof}
Let $G'$ be the graph that results from simplifying $G$ over edge $(u, v)$ and let $w$ be the vertex in $G'$ that replaces $u$ and $v$. Any Hamiltonian cycle $x_1, x_2, \ldots, x_{n-2}, u, v$ in $G$ using edge $(u, v)$ corresponds with Hamiltonian cycle $x_1, x_2, \ldots, x_{n-2}, w$ in $G'$. And any Hamiltonian cycle $x_1, x_2, \ldots, x_{n-2}, w$ in $G'$ corresponds with Hamiltonian cycle $x_1, x_2, \ldots, x_{n-2}, u, v$ in $G$ using edge $(u, v)$. Thus there is a bijection between Hamiltonian cycles of $G'$ and Hamiltonian cycles of $G$ using edge $(u, v)$.

But if either $u$ has outdegree $1$ or $v$ has indegree $1$, then every Hamiltonian cycle in $G$ must use edge $(u, v)$, and so the Hamiltonian cycles of $G$ using edge $(u, v)$ are all the Hamiltonian cycles of $G$. Thus there is a bijection between Hamiltonian cycles of $G'$ and Hamiltonian cycles of $G$, and so the numbers of Hamiltonian cycles in $G$ and $G'$ are the same. As desired, $G'$ is Hamiltonian if and only if $G$ is.
\end{proof}

\begin{lemma}
\label{lemma:simplify_nonalternating}
If $(u, v)$ is an edge of planar non-alternating directed graph $G$, then simplifying $G$ over $(u, v)$ maintains the planar non-alternating property of $G$. 
\end{lemma}

\begin{proof}
Let $G'$ be the graph that results from simplifying $G$ over edge $(u, v)$. Starting with a planar non-alternating embedding of $G$, the corresponding planar embedding of $G'$ will also be non-alternating. We prove this below.

If $x$ is a vertex of $G$ that is not $u$ or $v$, then in the planar non-alternating embedding $x$ will have all the inward oriented edges in one contiguous section. The simplification of $G$ over $(u, v)$ will at most affect $x$ by removing some edges incident on $x$. In no case does this introduce alternation of inward and outward oriented sections to $x$. Thus $x$ is non-alternating in the planar embedding of $G'$.

If $x$ is the new vertex introduced due to the simplification of $G$ over $(u, v)$, then $x$ is non-alternating in the planar embedding of $G'$ because (1) the inward oriented edges are all inherited from $u$, (2) the outward oriented edges are all inherited from $v$, and (3) the edges inherited from the two vertices by $x$ can be separated into two contiguous sections.

As desired, this shows that $G'$ is planar non-alternating.
\end{proof}

We apply these lemmas to prove that the Planar Non-Alternating Indegree-$2$ Outdegree-$2$ Hamiltonicity Problem is NP-hard:

\begin{theorem}
The Planar Non-Alternating Indegree-$2$ Outdegree-$2$ Hamiltonicity Problem is NP-hard.
\end{theorem}

\begin{proof}
We prove this via the following reduction from the Planar Max-Degree-$3$ Hamiltonicity Problem. On input a planar graph $G$ with all in- and out-degrees $1$ or $2$, repeatedly identify edges $(u, v)$ such that either $u$ has outdegree $1$ or $v$ has indegree $1$ and simplify $G$ over $(u,v)$. Only stop once no such edges $(u, v)$ can be found, at which point output the resulting graph $G'$.

First note that this algorithm runs in polynomial time since (1) simplification is a polynomial-time operation and (2) the number of simplifications of $G$ is bounded above by the number of vertices in $G$ since each simplification decreases the number of vertices by $1$. 

Suppose the input instance $G$ is a ``no'' instance of the Planar Max-Degree-$3$ Hamiltonicity Problem. This means that $G$ is not Hamiltonian. By repeated application of Lemma~\ref{lemma:simplify_hamiltonicity}, $G'$ is Hamiltonian if and only if $G$ is Hamiltonian. Thus $G'$ is not Hamiltonian and so $G'$ is a ``no'' instance of the Planar Non-Alternating Indegree-$2$ Outdegree-$2$ Hamiltonicity Problem. 

On the other hand, suppose the input instance $G$ is a ``yes'' instance of the Planar Max-Degree-$3$ Hamiltonicity Problem. By repeated application of Lemma~\ref{lemma:simplify_hamiltonicity}, $G'$ is Hamiltonian if and only if $G$ is Hamiltonian, so $G'$ must have a Hamiltonian cycle. Below we show that all in- and out-degrees in $G'$ are $2$ and that $G'$ is a planar non-alternating graph. Together, this is enough to imply that $G'$ is a ``yes'' instance of the Planar Non-Alternating Indegree-$2$ Outdegree-$2$ Hamiltonicity Problem. 

Since $G'$ has a Hamiltonian cycle, no vertex in $G'$ can have in- or out-degree $0$. Furthermore, no vertex in $G'$ can have in- or out-degree $1$ because the reduction does not stop simplifying the graph until there are no in- or out-degree $1$ vertices left. Thus every in- or out-degree in $G'$ is at least $2$. When simplifying a graph over an edge, every in- or out-degree in the resulting graph is less than or equal to some in- or out-degree in the initial graph. By repeatedly applying this rule, we see that every in- and out-degree in $G'$ is at most the largest in- or out-degree in $G$. But as $G$ is a Planar Max-Degree-$3$ Hamiltonicity instance, the largest in- or out-degree in $G$ is at most $2$. Thus, we can conclude that every in- and out-degree in $G'$ must be exactly $2$.

By repeated application of Lemma~\ref{lemma:simplify_nonalternating}, we know that provided the original graph $G$ is a planar non-alternating graph, the final graph $G'$ will be as well. But if $G$ is a planar max-degree-$3$ graph, then every vertex in $G$ is non-alternating in any planar embedding (since alternating vertices always have total degree at least $4$). Thus, any planar embedding of $G$ is a planar non-alternating embedding. We can therefore conclude that both $G$ and $G'$ are planar non-alternating graphs.

As desired, $G$ is a ``yes'' instance of the Planar Max-Degree-$3$ Hamiltonicity Problem if and only if $G'$ is a ``yes'' instance of the Planar Non-Alternating Indegree-$2$ Outdegree-$2$ Hamiltonicity Problem. Together with the fact that the reduction runs in polynomial time, we have our desired result: the Planar Non-Alternating Indegree-$2$ Outdegree-$2$ Hamiltonicity Problem is NP-hard.
\end{proof}

\subsection{Reduction to Planar $(\{k\}, \{4\})$-TRVB for any $k \ge 4$}
\label{section:hardness_1_b}

Consider the following algorithm $R_k$:

\begin{definition}
For $k \ge 4$, algorithm $R_k$ takes as input a planar non-alternating graph $G$ whose vertex in- and out-degrees all equal $2$, and outputs an instance $M'$ of Planar $(\{k\}, \{4\})$-TRVB.

To begin, we construct a labeled undirected multigraph $M$ as follows;
refer to Figure~\ref{figure:reduction_planar_example}.

First we build all the vertices (and vertex labels) of $M$. For each vertex in $G$, we include an unbreakable vertex in $M$ and for each edge in $G$ we include a breakable vertex in $M$. If $v$ is a vertex or $e$ is an edge of $G$, we define $m(v)$ and $m(e)$ to be the corresponding vertices in $M$.

Next we add all the edges of $M$. Fix vertex $v$ in $G$. Let $(u_1, v)$ and $(u_2, v)$ be the edges into $v$ and let $(v, w_1)$ and $(v, w_2)$ be the edges out of $v$. Then add the following edges to $M$:
\begin{itemize}
\item Add an edge from $m(v)$ to each of $m((u_1, v))$, $m((u_2, v))$, $m((v, w_1))$, and $m((v, w_2))$.
\item Add an edge from $m((v, w_1))$ to $m((v, w_2))$.
\item Add $k-3$ edges from $m((u_1,v))$ to $m((u_2,v))$.
\end{itemize}

Finally, pick any specific vertex $\hat{v}$ in $G$; refer to Figure~\ref{figure:reduction_planar_example_M'}. Let $(u_1, \hat{v})$ and $(u_2, \hat{v})$ be the edges into $\hat{v}$ and let $(\hat{v}, w_1)$ and $(\hat{v}, w_2)$ be the edges out of $\hat{v}$. We modify $M$ by removing vertex $m(\hat{v})$ (and all incident edges), and adding the two edges $(m((u_1, \hat{v})), m((u_2, \hat{v})))$, and $(m((\hat{v}, w_1)), m((\hat{v},w_2)))$. Call the resulting multigraph $M'$ and return it as output of algorithm $R_k$.
\end{definition}

In order to analyze the behavior of algorithm $R_k$, it will be helpful to have the following definition:

\begin{definition}
We say that two edges in a planar non-alternating indegree-$2$ outdegree-$2$ graph are \emph{conflicting} if they start or end at the same vertex. A Hamiltonian cycle in such a graph must contain exactly one out of every pair of conflicting edges.
\end{definition}

\begin{lemma}
The output $M'$ of $R_k$ is a planar labeled multigraph whose vertices are all breakable with degree $k$ or unbreakable with degree $4$.
\end{lemma}

\begin{proof}
Define all variables used in the description of $R_k$ as defined there. Because $G$ is planar non-alternating, we can immediately conclude that multigraph $M$ is planar as well (see Figure~\ref{figure:reduction_planar_example} for an example).

\begin{figure}[!htb]
    \centering
    \raisebox{-.5\height}{\includegraphics{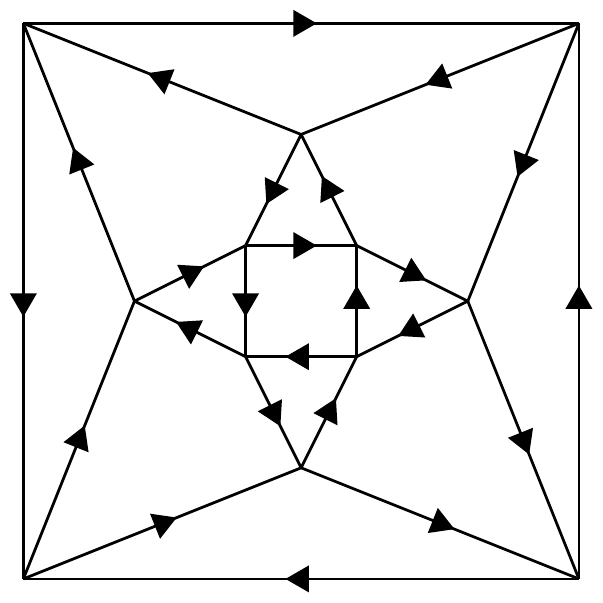}}~
    \raisebox{-.5\height}{\includegraphics{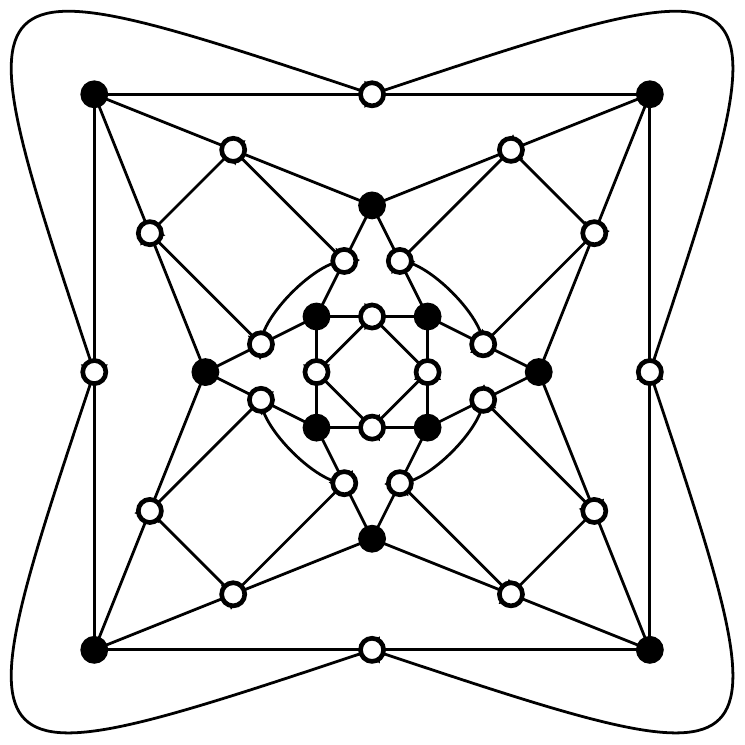}}
    \caption{If the planar non-alternating directed graph on the left is $G$, and if $k = 4$, then the produced $M$ is on the right. As you can see, planarity is maintained. If $k > 4$, then the output $M$ remains the same except some edges are duplicated; in that case too, $M$ is planar.}
    \label{figure:reduction_planar_example}
\end{figure}

Consider any vertex $m(v)$ in $M$ (where $v$ is a vertex of $G$). This vertex has exactly four neighbors: the vertices $m(e)$ for every edge $e$ in $G$ that is incident on $v$. Furthermore, this vertex is unbreakable.

Consider any vertex $m((u, v))$ in $M$ (where $(u, v)$ is an edge of $G$). This vertex has one edge to $m(u)$, one edge to $m(v)$, one edge to $m((u, v'))$, and $k-3$ edges to $m((u', v))$ (where $(u, v')$ and $(u', v)$ are the two edges in $G$ conflicting with $(u, v)$). Thus the degree of this vertex is $k$.

This shows that $M$ consists of only degree-$4$ unbreakable vertices and degree-$k$ breakable vertices. Thus, we have shown that $M$ has exactly those properties that we are trying to show for $M'$: $M$ is a planar labeled multigraph whose vertices are all breakable with degree $k$ or unbreakable with degree $4$. All that's left is to show that the operation converting $M$ to $M'$ leaves these properties unchanged.

To convert $M$ to $M'$, vertex $m(\hat{v})$ is removed, and two edges $(m((u_1, \hat{v})), m((u_2, \hat{v})))$, and $(m((\hat{v}, w_1)), m((\hat{v},w_2)))$ are added. 

Note first that the four endpoints of these two edges are exactly the four neighbors of $m(\hat{v})$ in $M$. Thus, each vertex in $M$ other than $m(\hat{v})$ has the same degree in $M'$: either the vertex was unaffected by the change from $M$ to $M'$ or a single edge was removed from the vertex and a single edge was added. Therefore the vertices of $M'$ are all breakable with degree $k$ or unbreakable with degree $4$. 

Next note that the two edges added to the multigraph are both already present. Increasing the multiplicity of an edge in a multigraph does not affect the planarity of the multigraph, and neither does removal of vertices and edges. Thus, the operation transforming $M$ into $M'$ maintains the planarity of the multigraph.

We can conclude that we have our desired properties: $M'$ is a planar labeled multigraph whose vertices are all breakable with degree $k$ or unbreakable with degree $4$. This can be seen for the Figure~\ref{figure:reduction_planar_example} example in Figure~\ref{figure:reduction_planar_example_M'}.

\begin{figure}[!htb]
    \centering
    \includegraphics{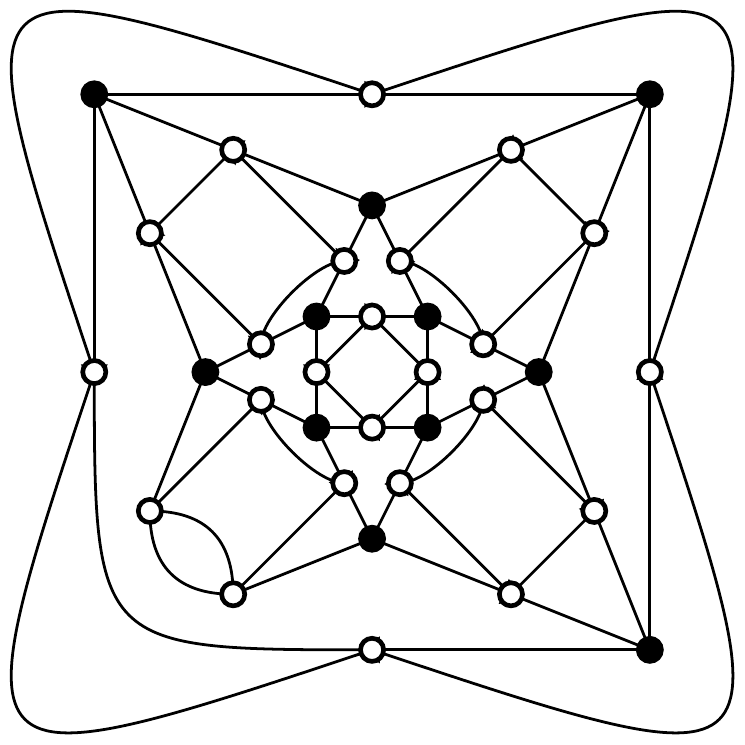}
    \caption{One possible $M'$ for the $M$ in Figure~\ref{figure:reduction_planar_example}, where $\hat{v}$ is chosen to be the bottom left vertex.  $M'$ is a planar labeled multigraph whose vertices are all breakable with degree $k$ or unbreakable with degree $4$.}
    \label{figure:reduction_planar_example_M'}
\end{figure}
\end{proof}

The following is an additional, trivial, property of $R_k$:

\begin{lemma}
$R_k$ runs in polynomial time.
\end{lemma}

Consider the following lemma:

\begin{lemma}
Suppose $R_k$ outputs $M'$ on input $G$ and there exists a solution to the TRVB problem on $M'$. Then the set of edges $e$ in $G$ such that $m(e)$ is not broken is a disjoint cycle cover of $G$.
\end{lemma}

\begin{proof}
Consider any pair of conflicting edges $e_1$ and $e_2$ in $G$ that share endpoint $v$. There exists at least one edge in $M$ between $m(e_1)$ and $m(e_2)$, and this edge is still present in $M'$. Thus, in order to avoid disconnecting that edge from the rest of the graph, either $m(e_1)$ or $m(e_2)$ must not be broken. $M$ also contains a cycle on three vertices $m(e_1)$, $m(e_2)$, and $m(v)$. If $v = \hat{v}$, then the third vertex is missing in $M'$, but in that case there is instead a cycle in $M'$ with just $m(e_1)$ and $m(e_2)$. In any case, $M'$ contains at least one cycle whose only breakable vertices are $m(e_1)$ and $m(e_2)$. In order for the resulting graph to be acyclic, at least one of these two vertices must be broken. This shows that in any solution to the TRVB problem on $M'$, exactly one out of every pair of conflicting edges $(e_1, e_2)$ has $m(e_i)$ broken.

Consider the set $C$ of edges $e$ in $G$ such that $m(e)$ is not broken. For every vertex $v$ of $G$, the two edges out of $v$ conflict and the two edges into $v$ conflict. Since every pair of conflicting edges $(e_1, e_2)$ has exactly one $m(e_i)$ broken, we conclude that $C$ contains one edge that enters $v$ and one that exists it. Thus $C$ is a disjoint cycle cover of $G$, as desired.
\end{proof}

Based on this lemma, we can define the following correspondence:

\begin{definition}
For any solution of TRVB instance $M'$, define $C$ to be the disjoint cycle cover of $G$ consisting of edges $e$ such that $m(e)$ is not broken in the given solution of $M'$.
\end{definition}

As per this definition, we can derive a disjoint cycle cover of $G$ from any solution to TRVB instance $M'$. Similarly, for any disjoint cycle cover of $G$, we can derive a candidate solution (though not necessarily an actual solution) for $M'$: simply break every vertex $m(e)$ where $e$ is an edge of $G$ that is not in the given disjoint cycle cover. Then for some suitable definition of ``candidate solution,'' there is a bijection between disjoint cycle covers of $G$ and candidate solutions of TRVB instance $M'$. We will show below that in fact, a disjoint cycle cover of $G$ is actually a Hamiltonian cycle if and only if the corresponding candidate solution for $M'$ is actually a solution. For example, see Figure~\ref{figure:reduction_planar_example_solutions}.

\begin{figure}[!htb]
    \centering
    \raisebox{-.5\height}{\includegraphics{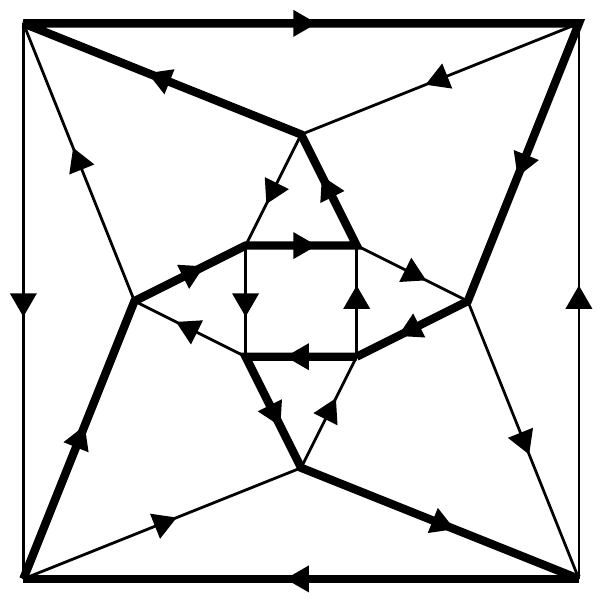}}~
    \raisebox{-.5\height}{\includegraphics{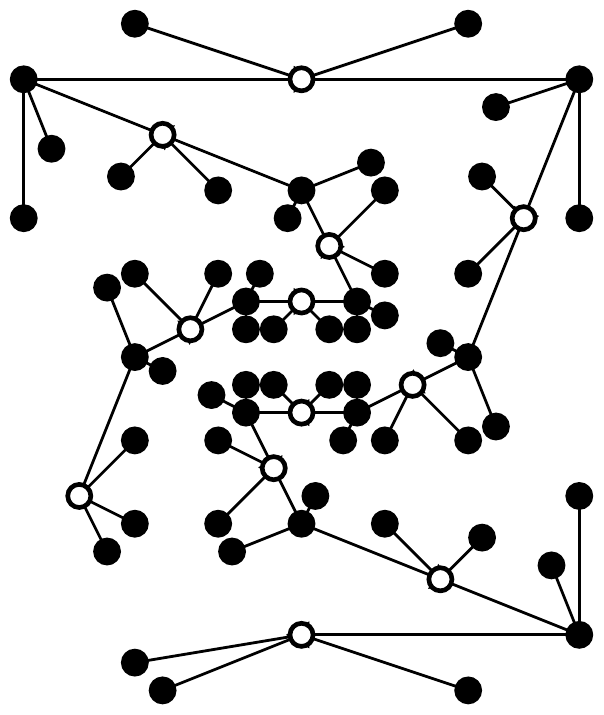}}
    \caption{This figure shows a Hamiltonian cycle in example graph $G$ from Figure~\ref{figure:reduction_planar_example} (left) and the corresponding solution of TRVB instance $M'$ shown in Figure~\ref{figure:reduction_planar_example_M'} (right).}
    \label{figure:reduction_planar_example_solutions}
\end{figure}

\begin{lemma}
Suppose $R_k$ outputs $M'$ on input $G$. If there exists a solution to the TRVB problem on $M'$, then the corresponding cycle cover of $G$ is actually a Hamiltonian cycle.
\end{lemma}

\begin{proof}
Let $C$ be the disjoint cycle cover of $G$ consisting of edges $e$ such that $m(e)$ is not broken in the given solution of $M'$. We know that $C$ is a union of disjoint cycles and we wish to show that there is exactly one cycle in $C$. Let $v$ be a vertex in $G$ and let $C_v$ be the cycle in $C$ containing $v$. We will prove that $C$ contains exactly one cycle by proving that $C_v$ contains every vertex of $G$.

Let $M'_{solved}$ be the solved version of $M'$ (after breaking vertices) and let $M_{solved}$ be a version of $M$ in which the same vertices are broken. Consider the difference between $M$ and $M'$ from a connectivity standpoint. In $M$, vertex $m(\hat{v})$ connects its four neighbors, while in $M'$, these neighbors are instead connected in pairs with two edges. Thus, $M$ is at least as connected as $M'$. This connectivity pattern carries through to the solved versions of these multigraphs: $M_{solved}$ is at least as connected as $M'_{solved}$. Since $M'_{solved}$ is a tree, it is fully connected, and so $M_{solved}$ is also fully connected. 

From this, we see that there exists a path in $M_{solved}$ from $m(v)$ to $m(v')$ for any vertex $v'$ in $G$. This path starts in $m(v)$, ends in $m(v')$, and passes through vertices that all have degree at least $2$. Therefore every vertex in this path is a vertex from the original multigraph $M$ that was not broken. We prove below that every path in $M_{solved}$ using only vertices originally in $M$ which starts in $m(v)$ must end at a vertex of the form $m(x)$ where $x$ is a vertex or edge in cycle $C_v$. Since there exists a path in $M_{solved}$ using only vertices originally in $M$ from $m(v)$ to $m(v')$, we conclude that $v'$ is a vertex in $C_v$. Applying this to every vertex in $G$, we see that $C_v$ is a cycle containing every vertex in $G$, and therefore $C = C_v$ is a Hamiltonian cycle.

Consider any path in $M_{solved}$ using only vertices originally in $M$ which starts in $m(v)$. We prove by induction on the path length that this path ends at a vertex of the form $m(x)$ where $x$ is a vertex or edge in cycle $C_v$.

If the path length is zero, then the end vertex is $m(v)$, which is certainly of the correct form.

Next, suppose that the statement holds for all paths of length $i-1$ or less. Then given a path of length $i$, we can apply the inductive hypothesis to this path without the last step. Thus we have that the pre-last node in the given path is of the form $m(x)$ where $x$ is a vertex or edge in cycle $C_v$. The final node in the path is a neighbor of $m(x)$ that is in $M$ and not a broken vertex. 

If $x$ is a vertex, then the only possible non-broken neighbors of $m(x)$ are the two nodes $m(e_1)$ and $m(e_2)$ where $e_1$ and $e_2$ are the two edges into and out of $x$ in $C_v$.

If $x$ is an edge, then the neighbors of $m(x)$ are nodes of the form $m(y)$ where $y$ is either a conflicting edge in $G$ or an endpoint of $x$. But since $m(x)$ is in $C_v$, it was not broken, which means that the vertices in $M$ corresponding to the conflicting edges were broken. Thus the only possible non-broken neighbors of $m(x)$ are the two nodes $m(e_1)$ and $m(e_2)$ where $e_1$ and $e_2$ are the two endpoints of $x$. Since $x$ is in $C_v$, so are its endpoints.

We conclude that in either case, the final node in the path is of the form $m(y)$ where $y$ is a vertex or edge in cycle $C_v$, proving the inductive step. By induction, any path in $M_{solved}$ using only vertices originally in $M$ which starts in $m(v)$ ends at a vertex of the form $m(x)$ where $x$ is a vertex or edge in cycle $C_v$.

As argued above, this implies that $C = C_v$ is a Hamiltonian cycle. Thus, we have shown that if the TRVB-problem $M'$ has a solution, then the corresponding cycle cover of $G$ is actually a Hamiltonian cycle.
\end{proof}

\begin{lemma}
Suppose $R_k$ outputs $M'$ on input $G$ and there exists a Hamiltonian cycle in $G$. Then the corresponding candidate solution of the TRVB instance $M'$ is a solution.
\end{lemma}

\begin{proof}

Suppose that the Hamiltonian cycle of $G$ is $C$. Then let $S$ be the set of vertices $m(e)$ such that $e$ is an edge of $G$ not in $C$. Let $M'_S$ be the result of breaking the vertices of $S$ in $M'$. Note that $M'_S$ is the candidate solution corresponding to cycle $C$. We will show below that $M'_S$ is a connected graph and that $M'_S$ has one fewer edges than it has vertices. In other words, $M'_S$ is a tree and so the candidate solution of the TRVB instance $M'$ corresponding to $C$ is an actual solution.

To begin, we show that $M'_S$ is connected. Let $X_0$ be the set of vertices in $M'_S$, let $X_1$ be the set of vertices in $M'_S$ that were in $M'$ before breaking the vertices of $S$, and let $X_2$ be the set of vertices in $M'_S$ of the form $m(v)$ for some vertex $v$ in $G$ with $v \ne \hat{v}$. We will show that (1) every vertex in $X_0 \setminus X_1$ is adjacent to a vertex in $X_1$ in graph $M'_S$, (2) every vertex in $X_1 \setminus X_2$ is adjacent to a vertex in $X_2$ in graph $M'_S$, and (3) there exists a path in $M'_S$ between any two vertices of $X_2$. Together, these three facts are sufficient to conclude that $M'_S$ is a connected graph.

The first fact we wish to show is that every vertex in $X_0 \setminus X_1$ is adjacent to a vertex in $X_1$ in graph $M'_S$. The vertices in $X_0 \setminus X_1$ are exactly the degree-$1$ vertices created when breaking the vertices of $S$ in $M'$. The vertices in $X_1$ are exactly the vertices in $M'_S$ that are originally in $M'$. Thus, we wish to show that every degree-$1$ vertex created by breaking vertices of $S$ in $M'$ ends up adjacent to a vertex originally in $M'$. This can fail to be the case only if two degree-$1$ vertices created by breaking vertices of $S$ in $M'$ end up adjacent to each other. This, in turn, is possible only if two vertices of $S$ are adjacent in $M'$. But if $m(e_1)$ and $m(e_2)$ are two vertices in $S$, then $e_1$ and $e_2$ cannot be conflicting edges (since out of every pair of conflicting edges exactly one is in $C$), and so there is no edge between $m(e_1)$ and $m(e_2)$. Thus $S$ cannot contain a pair of adjacent vertices, and so, as desired, every vertex in $X_0 \setminus X_1$ is adjacent to a vertex in $X_1$ in graph $M'_S$.

Next, we wish to show that every vertex in $X_1 \setminus X_2$ is adjacent to a vertex in $X_2$ in graph $M'_S$. The vertices in $X_1 \setminus X_2$ are exactly the vertices of the form $m(e)$ where $e$ is an edge in $C$. The vertices in $X_2$ are exactly the vertices of the form $m(v)$ where $v$ is a vertex of $G$ with $v \ne \hat{v}$. Then consider any vertex $m(e)$ in $X_1 \setminus X_2$ (where $e$ is an edge in $C$). Edge $e$ has two endpoints, so at least one of the two endpoints, call it $v$, is not equal to $\hat{v}$. Since $v$ is a vertex of $G$ with $v \ne \hat{v}$, we know that $m(v)$ is a vertex in $M'$ and furthermore, since $v$ is an endpoint of $e$, we know that $m(v)$ is adjacent to $m(e)$ in $M'$. Neither $m(v)$ nor $m(e)$ is in $S$, so the two remain adjacent in $M'_S$. Notice that $m(v)$ is a vertex in $X_2$, so as desired, $m(e)$, an arbitrary vertex in $X_1 \setminus X_2$, is adjacent to a vertex in $X_2$ in graph $M'_S$.

Finally, we wish to show that there exists a path in $M'_S$ between any two vertices of $X_2$. Vertices in $X_2$ are of the form $m(v)$ where $v \ne \hat{v}$ is a vertex in $G$. Thus, let $v_1$ and $v_2$ be two vertices in $G$ other than $\hat{v}$. We will demonstrate a path in $M'_S$ between $m(v_1)$ and $m(v_2)$.  Consider the path $v_1 = x_1, x_2, \ldots, x_l = v_2$ from $v_1$ to $v_2$ in $G$ which is part of Hamiltonian cycle $C$ but does not pass through vertex $\hat{v}$. Then consider the following list of vertices: 
$$m(x_1), m((x_1, x_2)), m(x_2), m((x_2, x_3)), m(x_3), \ldots, m((x_{l-1}, x_l)), m(x_l).$$
This list of vertices is a path in $M$, so since $m(\hat{v})$ is not in the list, this list is also a path in $M'$. Thus we have a path in $M'$ from $m(v_1)$ to $m(v_2)$. 

As described above, this allows us to conclude that $M'_S$ is connected. Next we will show that $M'_S$ has one fewer edge than it has vertices.

Suppose $G$ has $n$ vertices. Then the number of edges in $G$ is $2n$. The number of edges of $G$ not in $C$ is $n$, so $|S| = n$. Then the number of vertices in $M$ is $n + 2n = 3n$, the total number of vertices and edges in $G$. The number of edges in $M$ is $\frac{n \times 4 + 2n \times k}{2} = n(k+2)$. Transitioning from $M$ to $M'$ requires converting one vertex and four edges into zero vertices and two edges. Thus $M'$ has $3n-1$ vertices and $n(k+2) - 2$ edges. Each vertex in $S$ has degree $k$, so breaking the $n$ vertices of $S$ in $M'$ increases the number of vertices in the resulting multigraph by $n(k-1)$. Thus $M'_S$ has $3n-1 + n(k-1) = n(k+2)-1$ vertices and $n(k+2)-2$ edges. As desired, $M'_S$ has one fewer edge than it has vertices.

We showed above that $M'_S$ is connected and has one fewer edge than it has vertices so we can conclude that it is a tree. We have shown that if $G$ has a Hamiltonian cycle, breaking the vertices in $M'$ of set $S$ as defined above yields a tree. Thus we have that in the case that $G$ has a Hamiltonian cycle, the corresponding candidate solution of the TRVB instance $M'$ is a solution.
\end{proof}

\begin{theorem}
\label{theorem:planar_breakable_k_unbreakable_4}
Planar $(\{k\}, \{4\})$-TRVB is NP-hard for any $k \ge 4$.
\end{theorem}

\begin{proof}
Consider the following reduction from Planar Non-Alternating Indegree-$2$ Outdegree-$2$ Hamiltonicity Problem to Planar $(\{k\}, \{4\})$-TRVB. On input a graph $G$, we first check whether $G$ is a planar non-alternating graph all of whose in- and out-degrees are $2$. If yes, we run $R_k$ on input $G$ and output the result. Otherwise, we simply output any ``no'' instance of Planar $(\{k\}, \{4\})$-TRVB. 

Since $R_k$ runs in polynomial time, the above is clearly a polynomial-time reduction. Furthermore, $R_k$ always outputs a planar labeled multigraph whose vertices are all breakable with degree $k$ or unbreakable with degree $4$. As a result, in order to show that the above reduction is answer-preserving, it is sufficient to show that for all planar non-alternating graphs $G$ whose in- and out-degrees are $2$, $G$ has a Hamiltonian cycle if and only if the corresponding output $M'$ of $R_k$ on input $G$, when interpreted as a TRVB instance, has a solution. This is exactly what we showed in the previous two lemmas.

Since the Planar Non-Alternating Indegree-$2$ Outdegree-$2$ Hamiltonicity Problem is NP-hard, we conclude that for any $k \ge 4$, Planar $(\{k\}, \{4\})$-TRVB is NP-hard.
\end{proof}

\section{Planar TRVB and TRVB are NP-complete with high-degree breakable vertices}
\label{section:hardness_2}

The goal of this section is to show that Planar $(B, U)$-TRVB and $(B, U)$-TRVB are NP-complete if $B$ contains any $k \ge 4$. To do this, we will show that Planar $(\{k\}, \emptyset)$-TRVB is NP-hard for any $k \ge 4$.

\begin{lemma}
For any $k \ge 4$, there exists a reduction from either Planar $(\{k\}, \{3\})$-TRVB or Planar $(\{k\}, \{4\})$-TRVB to Planar $(\{k\}, \emptyset)$-TRVB.
\label{lemma:reduction_without_unbreakable_exists}
\end{lemma}

\begin{proof}
Below we will show that if $k \ge 4$, it is possible to simulate either a degree-$4$ unbreakable vertex or a degree-$3$ unbreakable vertex with a small gadget consisting of degree-$k$ breakable vertices. As a result, for every $k \ge 4$, we can construct a reduction from either Planar $(\{k\}, \{4\})$-TRVB or Planar $(\{k\}, \{3\})$-TRVB to Planar $(\{k\}, \emptyset)$-TRVB. 

In particular, our reduction simply replaces every occurrence of an unbreakable degree-$3$ or degree-$4$ vertex with the corresponding gadget made of breakable degree-$k$ vertices. Provided we can design gadgets of constant size (with respect to the size of $G$, not with respect to $k$) whose behavior is the same as the vertex they are simulating, this reduction will be correct and will run in polynomial time.

To design the gadgets, we have two cases.

For $k = 4$, we use the gadget shown in Figure~\ref{figure:unbreakable_4_from_breakable_4}. Suppose that the gadget shown was included in a Tree-Residue Vertex-Breaking instance. In order to break the cycle between $P_0$ and $P_1$ without disconnecting the edges between them from the rest of the graph, exactly one of those two vertices $P_i$ must be broken. But then the neighbor $Q_i$ of $P_i$ cannot be broken without disconnecting the edge $(P_i, Q_i)$ from the rest of the graph. On the other hand, the node $Q_{1-i}$ cannot be broken either since breaking it would disconnect $P_{1-i}$ from the rest of the graph. Thus any valid solution of the Tree-Residue Vertex-Breaking instance must break neither $Q_i$ and exactly one $P_i$. The resulting graph connects the other four neighbors of the $Q_i$s without forming any cycles. In other words the behavior of this gadget in a graph is the same as the behavior of an unbreakable degree-$4$ vertex.

\begin{figure}[!htb]
    \centering
    \includegraphics{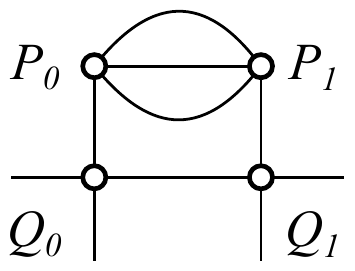}
    \caption{A gadget simulating an unbreakable degree-$4$ vertex using a planar arrangement of only breakable degree-$4$ vertices arranged.}
    \label{figure:unbreakable_4_from_breakable_4}
\end{figure}


For $k \ge 5$, we use the gadget shown in Figure~\ref{figure:unbreakable_?_from_breakable_k}. In this gadget, breakable vertex $Q$ has $2a$ edges to other vertices $P_0, P_1, \ldots, P_{2a-1}$ in the gadget and $k-2a$ edges out of the gadget. In addition, there are $k-1$ edges between $P_{2i}$ and $P_{2i+1}$ for every $i$ in $\{0, 1, \ldots, a-1\}$. Note that the degree of each vertex is $k$, as desired. When solving a graph containing this gadget, the cycle between $P_{2i}$ and $P_{2i+1}$ guarantees that exactly one of the two vertices must be broken. In order to not disconnect the other vertex from the rest of the graph, $Q$ cannot be broken in any valid solution. Thus, provided $a > 0$, every valid solution will break exactly one $P_{2i+j}$ with $j \in \{0,1\}$ for each $i$ and will not break $Q$. If this is done, the part of the resulting graph corresponding to this gadget will connect the $k-2a$ external neighbors of $Q$ to each other without forming any cycles. In other words the behavior of this gadget in a graph is the same as the behavior of an unbreakable degree-$(k-2a)$ vertex. 

\begin{figure}[!htb]
    \centering
    \includegraphics{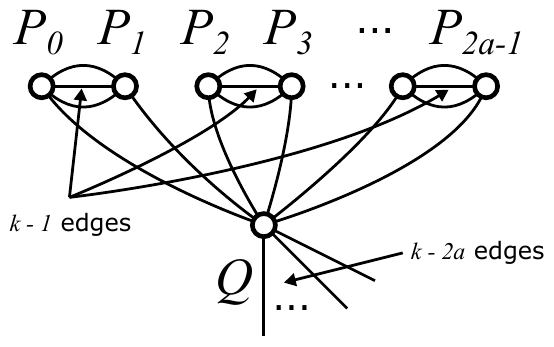}
    \caption{A gadget simulating an unbreakable degree-$(k-2a)$ vertex using only breakable degree-$k$ vertices arranged in a planar manner.}
    \label{figure:unbreakable_?_from_breakable_k}
\end{figure}


Since $k \ge 5$, it is possible to choose $a > 0$ such that $k-2a \in \{3, 4\}$. Then for every $k$, we are able to make a gadget to simulate either an unbreakable degree-$4$ vertex or an unbreakable degree-$3$ vertex. In all cases we can make the required gadgets, and so the reductions go through.
\end{proof}

We already know that Planar $(\{k\}, \{4\})$-TRVB is NP-hard from Section~\ref{section:hardness_1}, so to obtain NP-hardness from the previous lemma, all that is left is to show that Planar $(\{k\}, \{3\})$-TRVB is NP-hard.

\begin{lemma}
\label{lemma:planar_breakable_k_unbreakable_3}
Planar $(\{k\}, \{3\})$-TRVB is NP-hard for any $k \ge 4$.
\end{lemma}

\begin{proof}
We reduce from Planar $(\{k\}, \{4\})$-TRVB to Planar $(\{k\}, \{3\})$-TRVB. 

Consider any unbreakable degree-$4$ vertex $v$. We can replace $v$ with two new degree-$3$ unbreakable vertices $u$ and $u'$ with edges between the two new vertices and the neighbors of $v$ and an edge between $u$ and $u'$. Note that if we allocate two neighbors of $v$ to each of $u$ and $u'$, we succeed in making $u$ and $u'$ have degree $3$. Also note that it is possible to do this while maintaining the planarity of a multigraph. See Figure~\ref{figure:unbreakable_4_from_unbreakable_3}.

\begin{figure}[!htb]
    \centering
    \includegraphics{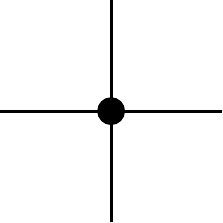}~
    \includegraphics{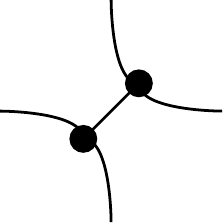}
    \caption{The degree-$4$ unbreakable vertex on the left can be simulated with two degree-$3$ unbreakable vertices as shown on the right while maintaining planarity.}
    \label{figure:unbreakable_4_from_unbreakable_3}
\end{figure}

Note that this pair of vertices ``behaves'' exactly the same as the original vertex did; in other words this change does not affect the answer to the Tree-Residue Vertex-Breaking question. As a result, applying this change to every unbreakable degree-$4$ vertex $v$ converts a Planar $(\{k\}, \{4\})$-TRVB instance into a Planar $(\{k\}, \{3\})$-TRVB instance.

By Theorem~\ref{theorem:planar_breakable_k_unbreakable_4}, Planar $(\{k\}, \{4\})$-TRVB is NP-hard for any $k \ge 4$, and so the existence of this reduction proves that Planar $(\{k\}, \{3\})$-TRVB is NP-hard for any $k \ge 4$.
\end{proof}

\begin{corollary}
Planar $(\{k\}, \emptyset)$-TRVB is NP-hard for any $k \ge 4$.
\label{corollary:planar_breakable_k}
\end{corollary}

\begin{proof}
Lemmas~\ref{lemma:reduction_without_unbreakable_exists} and~\ref{lemma:planar_breakable_k_unbreakable_3}, together with Theorem~\ref{theorem:planar_breakable_k_unbreakable_4}, allow us to conclude the desired result.
\end{proof}

\begin{theorem}
Planar $(B, U)$-TRVB is NP-complete if $B$ contains any $k \ge 4$. Also $(B, U)$-TRVB is NP-complete if $B$ contains any $k \ge 4$.
\end{theorem}

\begin{proof}
By Lemma~\ref{lemma:trivial_reductions}, there is a reduction from Planar $(\{k\}, \emptyset)$-TRVB to Planar $(B, U)$-TRVB if $B$ contains $k$. Thus, if $k \ge 4$ and $B$ contains $k$, then Planar $(B, U)$-TRVB is NP-hard. Lemma~\ref{lemma:trivial_reductions} also gives a reduction from Planar $(B, U)$-TRVB to $(B, U)$-TRVB, so we see that if $k \ge 4$ and $B$ contains $k$, then $(B, U)$-TRVB is also NP-hard. 

Using Corollary~\ref{corollary:np_membership}, we see that as desired, if $k \ge 4$ and $B$ contains $k$, then Planar $(B, U)$-TRVB and $(B, U)$-TRVB are both NP-complete.
\end{proof}

\section{Graph TRVB is NP-complete with high-degree breakable vertices}
\label{section:hardness_3}

The goal of this section is to show that Graph $(B, U)$-TRVB is NP-complete if $B$ contains any $k \ge 4$. To do this, we will show that Graph $(\{k\}, \emptyset)$-TRVB is NP-hard for any $k \ge 4$.

\begin{lemma}
\label{lemma:graph_breakable_k_unbreakable_2}
Graph $(\{k\}, \{2\})$-TRVB is NP-hard if $k \ge 4$.
\end{lemma}

\begin{proof}
In Corollary~\ref{corollary:planar_breakable_k} we saw that Planar $(\{k\}, \emptyset)$-TRVB is NP-hard if $k \ge 4$. We will reduce from this problem to Graph $(\{k\}, \{2\})$-TRVB.

In order to do so, we must convert a given multigraph into a graph. One way to do this is to insert two degree-$2$ unbreakable vertices into every edge. After doing this, there will no longer be any duplicated edges or self loops, and so the result will be graph. Furthermore, adding an unbreakable degree-$2$ vertex into the middle of an edge does not influence the answer to a Tree-Residue Vertex-Breaking question. Thus applying this transformation is a valid reduction.

We conclude that as desired, Graph $(\{k\}, \{2\})$-TRVB is NP-hard if $k \ge 4$.
\end{proof}

\begin{theorem}
\label{theorem:graph_breakable_k}
Graph $(\{k\}, \emptyset)$-TRVB is NP-hard if $k \ge 4$.
\end{theorem}

\begin{proof}
In Lemma~\ref{lemma:graph_breakable_k_unbreakable_2} we saw that Graph $(\{k\}, \{2\})$-TRVB is NP-hard if $k \ge 4$. We wish to reduce from that problem to Graph $(\{k\}, \emptyset)$-TRVB.

In order to do this, we construct a constant sized (in the size of $G$) gadget using degree-$k$ breakable vertices that behaves the same as a degree-$2$ unbreakable vertex. Simply replacing every degree-$2$ unbreakable vertex with a copy of this gadget is a valid reduction, allowing us to conclude that Graph $(\{k\}, \emptyset)$-TRVB is NP-hard if $k \ge 4$.

The gadget is shown in Figure~\ref{figure:unbreakable_2_from_breakable_k}. The gadget consists of $2k-2$ breakable vertices, each of degree $k$. Call them $P_1, P_2, \ldots, P_{k-2}$ and $Q_1, Q_2, \ldots, Q_k$. The gadget contains an edge between each pair $(P_i, Q_j)$ and an edge between each pair $(Q_i, Q_{i+1})$. Finally, both $Q_1$ and $Q_k$ will have one edge leaving the gadget. 

\begin{figure}[!htb]
    \centering
    \includegraphics{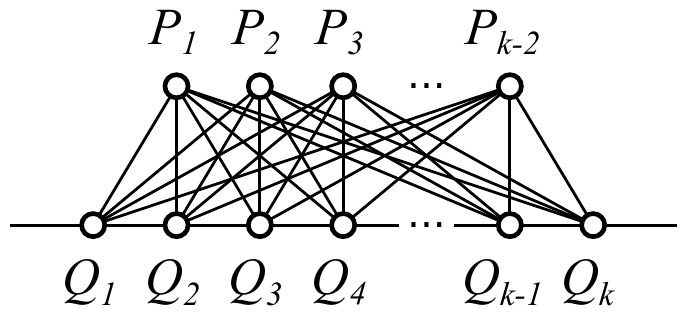}
    \caption{A gadget simulating an unbreakable degree-$2$ vertex using only breakable degree-$k$ vertices arranged without self loops or duplicated edges.}
    \label{figure:unbreakable_2_from_breakable_k}
\end{figure}

In a solution to this gadget, either $P_1$ is broken or not. If $P_1$ is broken, then to avoid disconnecting the edge $(Q_i, P_1)$ from the rest of the graph, $Q_i$ must not be broken. But $(Q_1, Q_2, P_i)$ is a cycle for every $i$, so in order to avoid having that cycle in the final graph, $P_i$ must also be broken.

If $P_1$ is not broken, then either $Q_1$ or $Q_2$ must be broken (due to cycle $(Q_1, Q_2, P_i)$). Then if $Q_i$ is broken, $P_j$ must not be broken in order to avoid disconnecting edge $(Q_i, P_j)$ from the rest of the graph. This means that every $P_j$ will not be broken. In that case, however, the existence of cycle $(Q_{i_1}, P_1, Q_{i_2}, P_2)$ guarantees that either $Q_{i_1}$ or $Q_{i_2}$ will be broken for every pair $i_1, i_2$. In other words, at most one $Q_i$ can be unbroken. This means, however, that either both $Q_1$ and $Q_2$ or both $Q_3$ and $Q_4$ will be broken, in either case isolating an edge from the rest of the graph. Thus we see that this case is not possible.

We can conclude that the only solution to this gadget is to break all of the $P_i$s but none of the $Q_i$s, thereby connecting the external neighbors of $Q_1$ and $Q_k$ (through the path of $Q_i$s) without leaving any cycles or disconnecting the graph. In other words, this gadget behaves like an unbreakable degree-$2$ vertex, as desired.

Thus we see that the reduction goes through and Graph $(\{k\}, \emptyset)$-TRVB is NP-hard if $k \ge 4$.
\end{proof}

\begin{corollary}
Graph $(B, U)$-TRVB is NP-complete if $B$ contains any $k \ge 4$. 
\end{corollary}

\begin{proof}
We saw in Theorem~\ref{theorem:graph_breakable_k} that Graph $(\{k\}, \emptyset)$-TRVB is NP-hard if $k \ge 4$, and we saw in Lemma~\ref{lemma:trivial_reductions}, there is a reduction from Graph $(\{k\}, \emptyset)$-TRVB to Graph $(B, U)$-TRVB if $B$ contains $k$. Thus, if $k \ge 4$ and $B$ contains $k$, then Graph $(B, U)$-TRVB is NP-hard. 
Using Corollary~\ref{corollary:np_membership}, we see that as desired, if $k \ge 4$ and $B$ contains $k$, then Graph $(B, U)$-TRVB is NP-complete.
\end{proof}

\section{Planar Graph TRVB is NP-hard with both low-degree vertices and high-degree breakable vertices}
\label{section:hardness_4}

The goal of this section is to show that Planar Graph $(B, U)$-TRVB is NP-complete if (1) either $B \cap \{1,2,3,4,5\} \ne \emptyset$ or $U \cap \{1,2,3,4\} \ne \emptyset$ and (2) there exists a $k \ge 4$ with $k \in B$. To do this, we will demonstrate that these conditions are sufficient to guarantee that it is possible to build small planar gadgets which behave like unbreakable degree-$2$ vertices. Inserting two copies of such a gadget into every edge converts a planar multigraph into a planar graph while keeping the answer to the TRVB question the same. This is a reduction from Graph $(B, U)$-TRVB to Planar Graph $(B, U)$-TRVB (provided both conditions (1) and (2) above hold).

Below, we prove the existence of the desired gadgets.

\begin{lemma}
There exists a planar gadget that simulates an unbreakable vertex of degree-$2$ built out of breakable degree-$k$ vertices (for any $k \ge 4$) and unbreakable degree-$4$ vertices such that the number of nodes is constant with respect to the size of a given multigraph $G$.
\end{lemma}

\begin{proof}
The gadget for this theorem is shown in Figure~\ref{figure:unbreakable_2_from_breakable_k_unbreakable_4}. For each breakable vertex in this figure, there exists a cycle in the gadget containing the vertex and no other breakable vertex. To avoid leaving this cycle in the final graph, the two breakable vertices in the gadget must both be broken in a valid solution. This fully determines the solution of the gadget. 

\begin{figure}[!htb]
    \centering
    \includegraphics{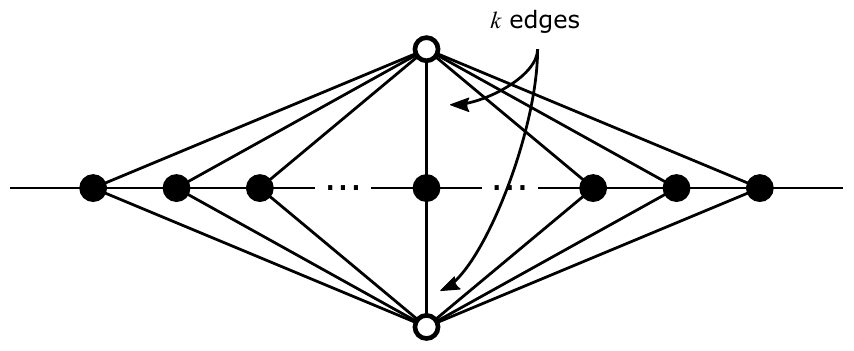}
    \caption{A gadget simulating an unbreakable degree-$2$ vertex using only breakable degree-$k$ and unbreakable degree-$4$ vertices arranged in a planar manner without self loops or duplicate edges.}
    \label{figure:unbreakable_2_from_breakable_k_unbreakable_4}
\end{figure}

Thus, if this gadget is included in a graph, the two breakable vertices must be broken, resulting in the gadget connecting the two edges that extend out to the rest of the graph. In other words, the gadget behaves like a degree-$2$ unbreakable vertex.

Note also that this gadget uses $k + 2$ nodes, which is constant with respect to the size of a given multigraph $G$.
\end{proof}

\begin{lemma}
There exists a planar gadget that simulates an unbreakable degree-$2$ vertex built out of breakable degree-$k$ vertices (for any $k \ge 4$) and unbreakable degree-$3$ vertices such that the number of nodes is constant with respect to the size of a given multigraph $G$.
\end{lemma}

\begin{proof}
The gadget for this theorem is shown in Figure~\ref{figure:unbreakable_2_from_breakable_k_unbreakable_3}. The one breakable vertex in this figure is in a cycle in the gadget (with no other breakable vertex). To avoid leaving this cycle in the final graph, the breakable vertex must be broken in a valid solution. This fully determines the solution of the gadget. 

\begin{figure}[!htb]
    \centering
    \includegraphics{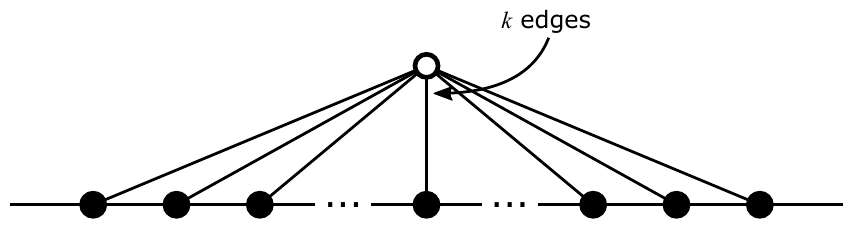}
    \caption{A gadget simulating an unbreakable degree-$2$ vertex using only breakable degree-$k$ and unbreakable degree-$3$ vertices arranged in a planar manner without self loops or duplicate edges.}
    \label{figure:unbreakable_2_from_breakable_k_unbreakable_3}
\end{figure}

Thus, if this gadget is included in a graph, the breakable vertex must be broken, resulting in the gadget connecting the two edges that extend out to the rest of the graph. In other words, the gadget behaves like a degree-$2$ unbreakable vertex.

Note also that this gadget uses $k + 1$ nodes, which is constant with respect to the size of a given multigraph $G$.
\end{proof}

\begin{lemma}
There exists a planar gadget that simulates an unbreakable degree-$2$ vertex built out of breakable degree-$k$ vertices (for any $k \ge 4$) and unbreakable degree-$1$ vertices such that the number of nodes is constant with respect to the size of a given multigraph $G$.
\end{lemma}

\begin{proof}
The gadget for this theorem is shown in Figure~\ref{figure:unbreakable_2_from_breakable_k_unbreakable_1}. The one breakable vertex in this figure cannot be broken (as that would separate the unbreakable vertices from the rest of the graph).

\begin{figure}[!htb]
    \centering
    \includegraphics{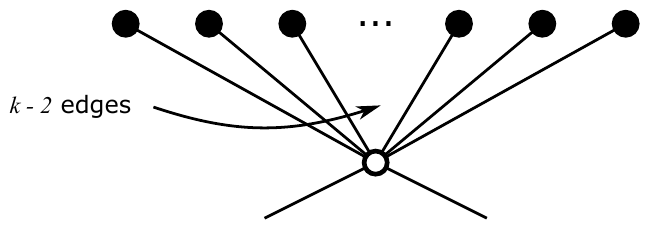}
    \caption{A gadget simulating an unbreakable degree-$2$ vertex using only breakable degree-$k$ and unbreakable degree-$1$ vertices arranged in a planar manner without self loops or duplicate edges.}
    \label{figure:unbreakable_2_from_breakable_k_unbreakable_1}
\end{figure}

Thus, if this gadget is included in a graph, the breakable vertex must not be broken, resulting in the gadget connecting the two edges that extend out to the rest of the graph. In other words, the gadget behaves like a degree-$2$ unbreakable vertex.

Note also that this gadget uses $k -1$ nodes, which is constant with respect to the size of a given multigraph $G$.
\end{proof}

\begin{lemma}
There exists a planar gadget that simulates an unbreakable degree-$2$ vertex built out of breakable degree-$k$ vertices (for any $k \ge 4$) and breakable degree-$1$ vertices such that the number of nodes is constant with respect to the size of a given multigraph $G$.
\end{lemma}

\begin{proof}
Breaking a degree-$1$ vertex does nothing, so breakable degree-$1$ vertices are essentially the same as unbreakable degree-$1$ vertices. Thus we can simply use the same construction as for the previous lemma.
\end{proof}

\begin{lemma}
There exists a planar gadget that simulates an unbreakable degree-$2$ vertex built out of breakable degree-$k$ vertices (for any $k \ge 4$) and breakable degree-$2$ vertices such that the number of nodes is constant with respect to the size of a given multigraph $G$.
\end{lemma}

\begin{proof}
We begin by constructing the gadget shown in Figure~\ref{figure:unbreakable_?_from_breakable_2_k}. In this gadget, breakable vertex $Q$ has $2a$ edges to other vertices $P_0, P_1, \ldots, P_{2a-1}$ in the gadget and $k-2a$ edges out of the gadget. In addition, there is an edge between $P_{2i}$ and $P_{2i+1}$ for every $i$ in $\{0, 1, \ldots, a-1\}$. Note that the degree of $Q$ is $k$ and the degree of each $P_i$ is $2$. When solving a graph containing this gadget, the cycle $(Q, P_{2i}, P_{2i+1})$ guarantees that exactly one of the three vertices in the cycle must be broken. $Q$, however, cannot be broken without disconnecting $P_{2i}$ and $P_{2i+1}$ from the rest of the graph. Thus, provided $a > 0$, every valid solution will break exactly one $P_{2i+j}$ out of every pair $(P_{2i}, P_{2i+1})$ and will not break $Q$. If this is done, the part of the resulting graph corresponding to this gadget will connect the $k-2a$ external neighbors of $Q$ to each other without forming any cycles. In other words the behavior of this gadget in a graph is the same as the behavior of an unbreakable degree-$(k-2a)$ vertex. 

\begin{figure}[!htb]
    \centering
    \includegraphics{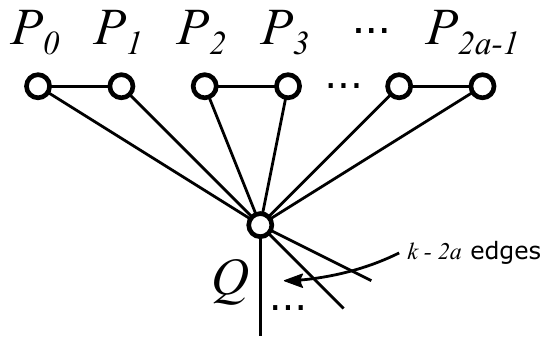}
    \caption{A gadget simulating an unbreakable degree-$(k-2a)$ vertex using only breakable degree-$k$ and degree-$2$ vertices arranged in a planar manner without self loops or duplicate edges.}
    \label{figure:unbreakable_?_from_breakable_2_k}
\end{figure}

Note that the number of nodes in the above gadget is constant with respect to the size of a given multigraph $G$ (in particular, there are $2a+1\le k+1$ nodes).

Since $k \ge 4$, we can select $a > 0$ such that $k-2a \in \{2, 3\}$. In other words, the above gadget behaves either as an unbreakable degree-$2$ vertex gadget or as an unbreakable degree-$3$ vertex gadget.

If the gadget behaves as an unbreakable degree-$3$ vertex gadget, then an unbreakable degree-$2$ vertex gadget can be built (as in a previous lemma) using breakable degree-$k$ vertices and unbreakable degree-$3$ vertex gadgets. In this case, the size of the new combined gadget is at most a constant times the size of the above gadget. 

Thus in all cases we can construct a gadget simulating an unbreakable degree-$2$ vertex using only degree-$k$ and degree-$2$ breakable vertices.
\end{proof}

\begin{lemma}
There exists a planar gadget that simulates an unbreakable degree-$2$ vertex built out of breakable degree-$3$ vertices such that the number of nodes is constant with respect to the size of a given multigraph $G$.
\end{lemma}

\begin{proof}
The gadget for this theorem is shown in Figure~\ref{figure:unbreakable_2_from_breakable_3}. If either vertex $P$ or vertex $Q_2$ is broken, then none of the others can be (since all the non-$P$ vertices are adjacent to $P$ and all the non-$Q_2$ vertices are adjacent to $Q_2$). If neither $P$ nor vertex $Q_2$ is broken, then in order to avoid having cycles, both $Q_1$ and $Q_3$ must be broken; this however, disconnects $P$ and $Q_2$ from the rest of the graph. Thus the only valid solutions of this gadget break exactly one of $P$ and $Q_2$ and nothing else. In either case, the resulting graph piece connects the two edges that extend out to the rest of the graph. In other words, the gadget behaves like a degree-$2$ unbreakable vertex.

\begin{figure}[!htb]
    \centering
    \includegraphics{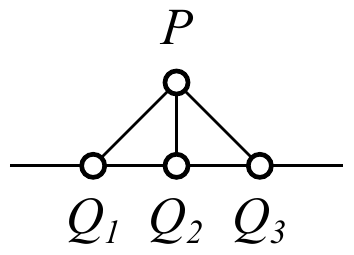}
    \caption{A gadget simulating an unbreakable degree-$2$ vertex using only breakable degree-$3$ vertices arranged in a planar manner without self loops or duplicate edges.}
    \label{figure:unbreakable_2_from_breakable_3}
\end{figure}

Note also that this gadget uses $4$ nodes, which is constant with respect to the size of a given multigraph $G$.
\end{proof}

\begin{lemma}
There exists a planar gadget that simulates an unbreakable degree-$2$ vertex built out of breakable degree-$4$ vertices such that the number of nodes is constant with respect to the size of a given multigraph $G$.
\end{lemma}

\begin{proof}
The gadget for this theorem is shown in Figure~\ref{figure:unbreakable_2_from_breakable_4}. Note that this is actually the same gadget as described in Theorem~\ref{theorem:graph_breakable_k} for $k = 4$. Thus we have already argued the correctness of this gadget.

\begin{figure}[!htb]
    \centering
    \includegraphics{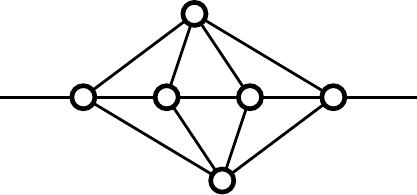}
    \caption{A gadget simulating an unbreakable degree-$2$ vertex using only breakable degree-$4$ vertices arranged in a planar manner without self loops or duplicate edges.}
    \label{figure:unbreakable_2_from_breakable_4}
\end{figure}

Note also that this gadget uses $6$ nodes, which is constant with respect to the size of a given multigraph $G$.
\end{proof}

\begin{lemma}
There exists a planar gadget that simulates an unbreakable degree-$2$ vertex built out of breakable degree-$5$ vertices such that the number of nodes is constant with respect to the size of a given multigraph $G$.
\end{lemma}

\begin{proof}
The gadget for this theorem is shown in Figure~\ref{figure:unbreakable_2_from_breakable_5}. 

\begin{figure}[!htb]
    \centering
    \includegraphics[width=\textwidth]{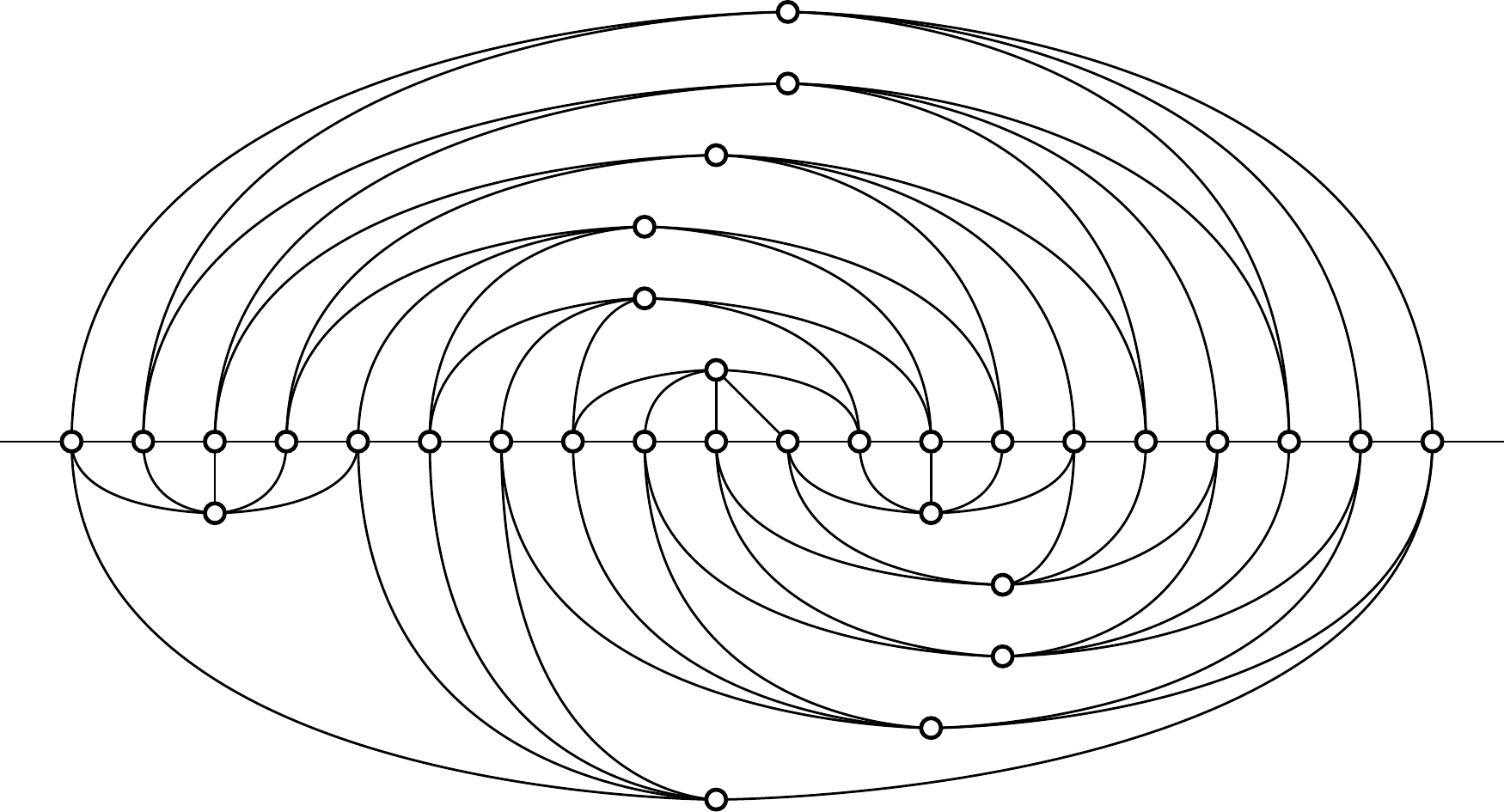}
    \caption{A gadget simulating an unbreakable degree-$2$ vertex using only breakable degree-$5$ vertices arranged in a planar manner without self loops or duplicate edges.}
    \label{figure:unbreakable_2_from_breakable_5}
\end{figure}

This gadget contains exactly thirty-two degree-$5$ vertices with two edges leaving the gadget. A choice of vertices to break within the gadget could be a valid solution if either (1) the resulting graph restricted to the vertices within the gadget is a tree (with the two edges extending out of the gadget connected to this tree) or (2) the resulting graph restricted to the vertices within the gadget consists of two trees (with the two edges extending out of the gadget connected to one of these trees each). Since there are $32$ degree-$5$ vertices with two edges out of the gadget, the number of edges in the gadget is $\frac{32\times5 - 2}{2} = 79$. Note that breaking vertices does not affect the number of edges in a graph, so the final tree or pair of trees in a valid solution of the gadget will also have $79$ edges. A tree with $79$ edges has $80$ vertices while two trees with $79$ edges have $81$ vertices. Breaking one vertex increases the number of vertices by $4$. Thus, it is possible to achieve the one-tree solution by breaking $\frac{80-32}{4} = 12$ vertices and it is impossible to achieve the two-tree solution.

The one-tree solution corresponds with the situation in which the gadget connects the two edges that extend out to the rest of the graph. In other words, provided that the gadget can be solved, every possible solution is one under which the gadget behaves like a degree-$2$ unbreakable vertex. There is, however, still the question of whether the gadget can be solved. In fact, breaking every vertex above or below the center horizontal line of the gadget (and leaving the 20 vertices along the center line unbroken) is a valid solution of the gadget.

Therefore this gadget behaves like a degree-$2$ unbreakable vertex.

Note also that this gadget uses $32$ nodes, which is constant with respect to the size of a given multigraph $G$.
\end{proof}

With these gadgets, we can now reduce from the Planar TRVB variants:

\begin{theorem}
Planar Graph $(B, U)$-TRVB is NP-complete if (1) either $B \cap \{1,2,3,4,5\} \ne \emptyset$ or $U \cap \{1,2,3,4\} \ne \emptyset$ and (2) there exists a $k \ge 4$ with $k \in B$.
\end{theorem}

\begin{proof}
Suppose that for some $k \ge 4$, we have that $k \in B$ and also that either $B \cap \{1,2,3,4,5\} \ne \emptyset$ or $U \cap \{1,2,3,4\} \ne \emptyset$. Then we can reduce from Planar $(\{k\}, \emptyset)$-TRVB to Planar Graph $(B, U)$-TRVB. Our reduction works by inserting either two degree-$2$ unbreakable vertices or two degree-$2$ unbreakable vertex gadgets (built out of vertices whose types are allowed in $(B, U)$-TRVB) into each edge. In either case, the resulting multigraph uses only vertices with allowed degrees (with the answer staying the same), but is now also a graph. 

There are several cases:

If $B \cap \{3,4,5\} \ne \emptyset$, then let $b$ be an element of $B \cap \{3,4,5\}$. We can build a constant size planar gadget which behaves like an unbreakable degree-$2$ vertex using only breakable degree-$b$ vertices. 

If $B \cap \{1,2\} \ne \emptyset$, then let $b$ be an element of $B \cap \{1,2\}$. We can build a constant size planar gadget which behaves like an unbreakable degree-$2$ vertex using only breakable degree-$b$ vertices and breakable degree-$k$ vertices.

If $U \cap \{1,3,4\} \ne \emptyset$, then let $u$ be an element of $U \cap \{1,3,4\}$. We can build a constant size planar gadget which behaves like an unbreakable degree-$2$ vertex using only unbreakable degree-$u$ vertices and breakable degree-$k$ vertices.

If $2 \in U$, then the problem we are reducing to allows degree-$2$ unbreakable vertices.

Thus, in all cases either (1) the problem we are reducing to allows degree-$2$ unbreakable vertices or (2) the problem we are reducing to allows vertex types with which we can build a constant sized gadget which behaves like a degree-$2$ unbreakable vertex. 

Thus, our desired reduction from Planar $(\{k\}, \emptyset)$-TRVB to Planar Graph $(B, U)$-TRVB is possible. Since Planar $(\{k\}, \emptyset)$-TRVB is NP-hard (by Corollary~\ref{corollary:planar_breakable_k}), we see that Planar Graph $(B, U)$-TRVB is NP-hard. By Corollary~\ref{corollary:np_membership}, Planar Graph $(B, U)$-TRVB is in NP, so as desired, Planar Graph $(B, U)$-TRVB is NP-complete.
\end{proof}

\section{Planar Graph TRVB is polynomial-time solvable without small vertex degrees}
\label{section:polynomial}

The overall purpose of this section is to show that variants of Planar Graph TRVB which disallow all small vertex degrees are polynomial-time solvable because the answer is always ``no.'' Consider for example the following theorem.

\begin{theorem}
If $b > 5$ for every $b \in B$ and $u > 5$ for every $u \in U$, then Planar Graph $(B,U)$-TRVB has no ``yes'' inputs. As a result, Planar Graph $(B,U)$-TRVB problem is polynomial-time solvable.
\end{theorem}

\begin{proof}
The average degree of a vertex in a planar graph must be less than $6$, so there are no planar graphs with all vertices of degree at least $6$. Thus, if $b > 5$ for every $b \in B$ and $u > 5$ for every $u \in U$, then every instance of Planar Graph $(B,U)$-TRVB is a ``no'' instance.
\end{proof}

In fact, we will strengthen this theorem below to disallow ``yes'' instances even when degree-$5$ unbreakable vertices are present by using the particular properties of the TRVB problem. Note that this time, planar graph inputs which satisfy the degree constraints are possible, but any such graph will still yield a ``no'' answer to the Tree-Residue Vertex-Breaking problem. 

We begin with the proof idea in Section~\ref{section:polynomial_2_proof_idea}, and proceed through the details in Section~\ref{section:polynomial_2_proof}

\subsection{Proof idea}
\label{section:polynomial_2_proof_idea}

Consider the hypothetical situation in which we have a solution to the TRVB problem in a planar graph whose unbreakable vertices each have degree at least $5$ and whose breakable vertices each have degree at least $6$. The general idea of the proof is to show that this situation is impossible by assigning a scoring function (described below) to the possible states of the graph as vertices are broken. The score of the initial graph can easily be seen to be zero and assuming the TRVB instance has a solution, the score of the final tree can be shown to be positive. It is also the case, however, that if we break the vertices in the correct order, no vertex increases the score when broken, implying a contradiction. 

Next, we introduce the scoring mechanism. Consider one vertex in the graph after some number of vertices have been broken. This vertex has several neighbors, some of which have degree $1$. We can group the edges of this vertex that lead to degree-$1$ neighbors into ``bundles'' seperated by the edges leading to higher degree neighbors. For example, in Figure~\ref{figure:bundles_example}, the vertex shown has two bundles of size $2$ and one bundle of size $3$. Each bundle is given a score according to its size, and the score of the graph is equal to the cumulative score of all present bundles. In particular, if a bundle has a size of $1$, then we assign the bundle a score of $-1$, and otherwise we assign the bundle a score of $n-1$ where $n$ is the size of the bundle. 

\begin{figure}[!htb]
    \centering
    \includegraphics{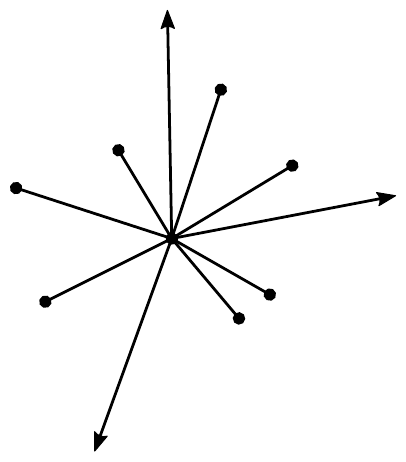}
    \caption{A degree-$10$ vertex with seven degree-$1$ neighbors (shown) and three other neighbors (not shown). The edges to the degree-$1$ neighbors form two bundles of size $2$ and one bundle of size $3$. }
    \label{figure:bundles_example}
\end{figure}

As it turns out, under this scoring mechanism, any tree all of whose non-leaves have degree at least $5$ always has a positive score. In fact, it is easy to see that in our TRVB instance, if breaking some set of breakable vertices $S$ results in a tree, then this degree constraint applies: the non-leaves are vertices from the original graph and therefore have degree at least $5$. Thus, the score of the original graph is zero (since there are no bundles), and the score after all the vertices in $S$ are broken is positive. 

Next, we define a breaking order for the vertices of $S$. In short, we will break the vertices of $S$ starting on the exterior of the graph and moving inward. More formally, we will repeatedly do the following step until all vertices in $S$ have been broken. Consider the external face of the graph at the current stage of the breaking process. Since not every vertex in $S$ has been broken, the graph is not yet a tree and the current external face is a cycle. Every cycle in the graph must contain a vertex from $S$ (in order for the final graph to be a tree), so choose a vertex from $S$ on the current external face and break that vertex next. 

Breaking the vertices of $S$ in this order has an interesting effect on the bundles in the graph: since every vertex from $S$ is on the external face when it is broken, every degree-$1$ vertex ends up within the external face when it appears. Thus all bundles are within the external face of the graph at all times.

Consider the effect that breaking one vertex from $S$ with degree $d \ge 6$ has on the score of the graph. Any vertex in $S$ on the external face has exactly two edges which border this face. The remaining $d-2$ edges must all leave the vertex into the interior of the graph. When the vertex is broken, each of these $d-2$ edges becomes a new bundle (since the interior of the graph never has any bundles). Thus, breaking the vertex creates $d-2$ new bundles of size $1$, thereby decreasing the score of the graph by $d-2$. On the other hand, the two edges which were on the external face are now each added to a bundle, thereby increasing the size of that bundle by one and increasing its score by at most two (in the case that the size was originally $1$). Thus, the increase in the score of the graph due to these two edges is at most $4$. In summary, breaking one vertex decreases the graph's score by $d-2 \ge 4$ and increases the graph's score by at most $4$. Thus, the total score of the graph does not increase.

Since the score of the graph does not increase with any step of the process, the final result should have at most the same score as the original graph. This contradicts the fact that the tree at the end of the process has positive score while the original graph has score zero. By contradiction, we conclude that $S$ cannot exist, giving us our desired result.

\subsection{Proof}
\label{section:polynomial_2_proof}

In this section, we will follow the proof outline given in the previous section. We begin with a sequence of definitions leading to a formal definition of the scoring function used above. 

Throughout this section, we will be considering a planar graph $G$ whose breakable vertices each have degree at least $6$ and whose unbreakable vertices each have degree at least $5$. 

\begin{definition} 
We say that $G'$ is a \emph{state} of $G$ if we can obtain $G'$ by breaking some vertices of $G$.
\end{definition}

We choose a particular planar embedding for $G$ to be the canonical planar embedding for $G$. When a vertex is broken in some state $G'$ of $G$, the new state $G''$ can inherit a planar embedding from $G'$ in the natural way: all vertices and edges unaffected by the vertex-breaking are embedded in the same place while the new vertices are placed so that the order of edges around each vertex is preserved. For example, breaking the top vertex in the planar embedding shown in the left part of Figure~\ref{figure:embedding_example} would yield the planar embedding shown in the center of the figure rather than the right part. Then every state of $G$ can (via a sequence of states) inherit the canonical planar embedding for $G$. We will then use this embedding as the canonical embedding for the state of $G$. With that done, we no longer have to specify which planar embedding we are using for the states of $G$: we will always use the canonical planar embedding.

\begin{figure}[!htb]
    \centering
    \hfill
    \includegraphics[scale=.5]{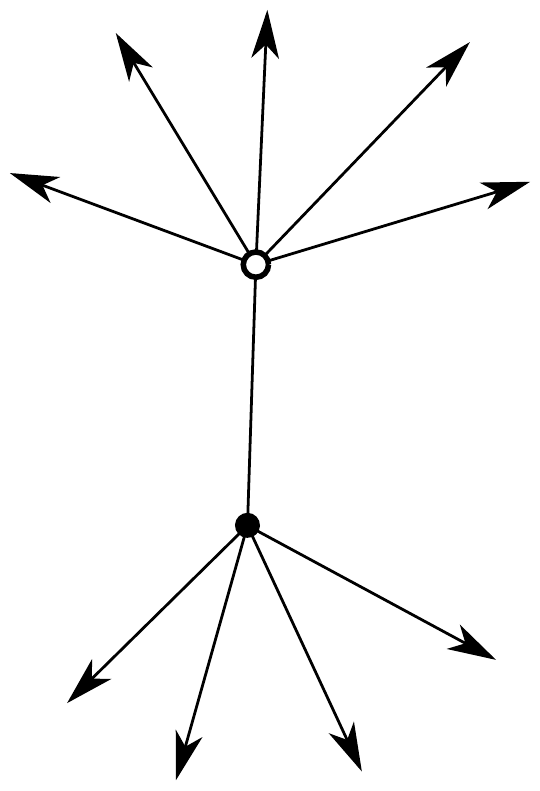}\hfill
    \includegraphics[scale=.5]{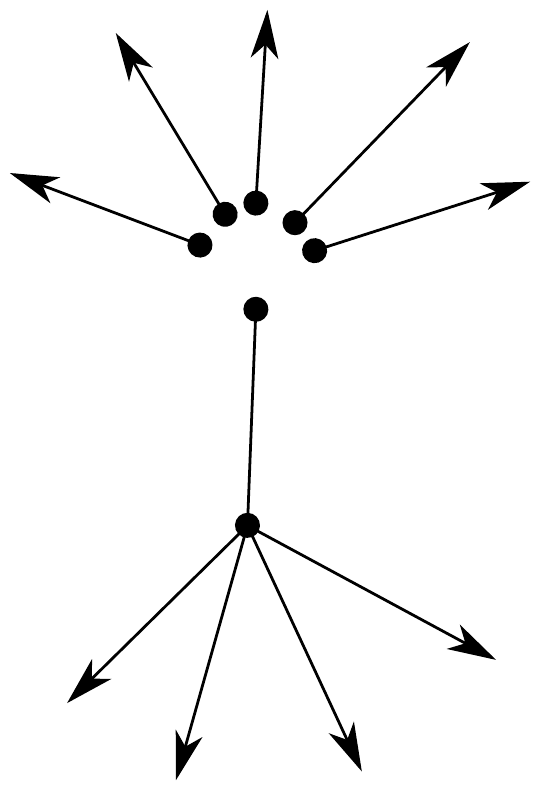}\hfill
    \includegraphics[scale=.5]{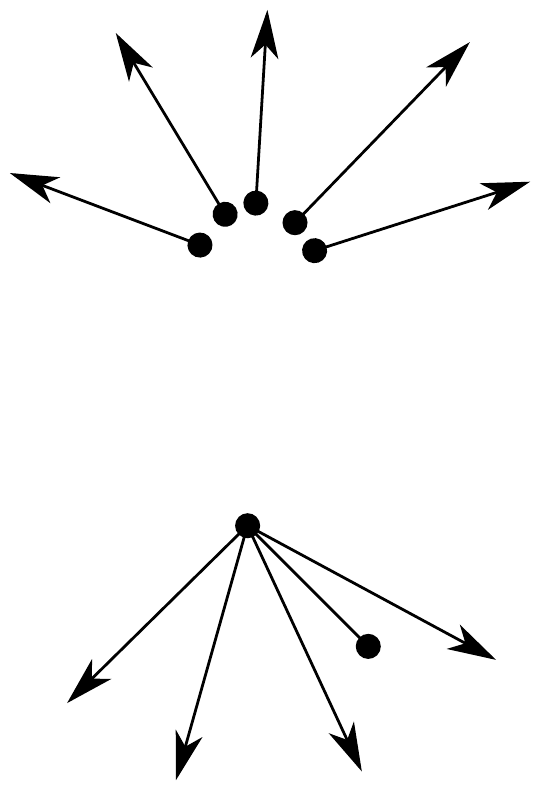}\hfill
    \hfill
    \caption{Breaking the top vertex in the first planar embedding should yield the second planar embedding rather than the third in order to maintain the order of edges around the bottom vertex.}
    \label{figure:embedding_example}
\end{figure}

\begin{definition}
We call $C$ a \emph{contiguous set of edges at $x$} if $C$ is a set of edges all sharing endpoint $x$ and we can proceed clockwise around $x$ starting and ending at some edge in $C$ such that the edges encountered are exactly those in $C$. Then a \emph{bundle} at vertex $x$ is a non-empty maximal contiguous set of edges at $x$ whose other endpoints have degree $1$.

Define the \emph{score of a bundle at vertex $x$} to be $-1$ if the bundle has size $1$ and $n-1$ if the bundle has size $n > 1$. Define the \emph{score} of $x$ to be the cumulative score of all the bundles at $x$. Define the \emph{score} of a state $G'$ of $G$ to be the cumulative score of all the vertices in $G'$.
\end{definition}

Suppose for the sake of contradiction that $S$ is a set of breakable vertices in $G$ such that breaking the vertices in $S$ yields state $T$ which is a tree. In other words, suppose that there exists a solution to the instance $G$ of Planar Graph $(\{6,7,8,\ldots\}, \{5,6,7,\ldots\})$-TRVB.

We begin by showing one side of the contradiction: that the score of $T$ is positive. To do this, we introduce the following lemma about trees:

\begin{lemma}
If $T'$ is a tree with at least $2$ vertices, then 
$$2(\text{number of leaves of}~T') + (\text{number of degree-$2$ vertices in}~T') > (\text{number of edges in}~T').$$ 
\end{lemma}

\begin{proof}
Define $n_1(T')$ to be the number of leaves in $T'$, define $n_2(T')$ to be the number of degree-$2$ vertices in $T'$, and define $n_e(T')$ to be the number of edges in $T'$. Then we wish to show that for any tree $T'$ with at least $2$ vertices, $2n_1(T') + n_2(T') > n_e(T')$. We will prove this by induction on the number of vertices in $T'$.

First consider the base case: if $X'$ is a tree on $2$ vertices, then $X'$ contains exactly one edge (between its two vertices) so the number of leaves of $X'$ is $2$, the number of degree-$2$ vertices is $0$ and the number of edges is $1$. We see then that $2n_1(X') + n_2(X') = 2\times2+0 = 4 > 1 = n_e(X')$ as desired.

Next suppose that for any tree $X$ on $i-1$ vertices it is the case that $2n_1(X) + n_2(X) > n_e(X)$. Let $X'$ be any tree on $i$ vertices and let $v$ be a leaf of $X'$. The graph $X' - \{v\}$ is a tree with $i-1$ vertices, so we can apply the inductive hypothesis: $2n_1(X' - \{v\}) + n_2(X' - \{v\}) > n_e(X' - \{v\})$.

Let $u$ be the sole neighbor of $v$ in $X$. The value $2n_1 + n_2$ (twice the number of leaves plus the number of degree-$2$ vertices) changes as we go from $X' - \{v\}$ to $X'$ due to the change in degree of $u$ and the addition of $v$. Adding a neighbor to $u$ either converts $u$ from being a leaf to a degree-$2$ vertex, converts $u$ from being a degree-$2$ vertex to a degree-$3$ vertex, or converts $u$ from being a vertex of degree $n > 2$ to a vertex of degree $n+1$. In all cases, the value $2n_1+n_2$ decreases by at most one due to the change in degree of $u$. The addition of the new leaf $v$, on the other hand, increases this value by $2$. Thus the overall increase of the value $2n_1+n_2$ when going from tree $X' - \{v\}$ to tree $X'$ is at least $1$. In other words, we have that $2n_1(X') + n_2(X') \ge 2n_1(X' - \{v\}) + n_2(X' - \{v\}) + 1$. Note also that the number of edges in $X'$ is one more than the number of edges in $X' - \{v\}$ (i.e. $n_e(X') = n_e(X' - \{v\}) + 1$).

Putting this all together, we see that 
$$2n_1(X') + n_2(X') \ge 2n_1(X' - \{v\}) + n_2(X' - \{v\}) + 1 > n_e(X' - \{v\}) + 1 = n_e(X').$$

As desired, we have shown for any tree $X'$ with $i$ vertices that $2n_1(X') + n_2(X') > n_e(X')$ concluding the inductive step. By induction, we have shown that for any tree $T'$ with at least $2$ vertices, $2n_1(T') + n_2(T') > n_e(T')$.
\end{proof}

\begin{lemma}
The score of $T$ is positive.
\end{lemma}

\begin{proof}
Every vertex in $T$ is either a vertex originally in $G$ or a new vertex created by the breaking of some vertex in $S$. Vertices in $G$ have degree at least $5$ and vertices created by the breaking of a vertex have degree $1$. Thus every vertex in $T$ that is not a leaf is a vertex originally in $G$. Let $T'$ be the tree formed by removing every leaf from $T$. Notice that the vertices in $T'$ are exactly the vertices in $G \setminus S$.

Since every tree has at least one vertex of degree at most $1$ and $G$ does not, we know that $G$ is not a tree. Thus $G \ne T$ and so $S \ne \emptyset$. Consider any vertex $v \in S$; $v$ has at least $6$ neighbors, none of which can be in $S$ (since then breaking $S$ disconnects the graph). Thus the number of vertices in $T'$ is at least $6$.

Below, we will show that the score of $T$ is 
$$-2(\text{number of edges in}~T') + 4(\text{number of leaves of}~T') + 2(\text{number of degree-$2$ vertices in}~T').$$
But since $T'$ is a tree with at least $6$ vertices, the previous lemma applies to show that $$2(\text{number of leaves of}~T') + (\text{number of degree-$2$ vertices in}~T') > (\text{number of edges in}~T').$$ 
Simply rearranging (and doubling) this inequality immediately shows that the score of $T$ is positive. Thus, all that is left is to show that the expression given above for the score of $T$ is correct.

The score of $T$ is the sum over all vertices $x$ in $T$ of the score of $x$. If $x$ is not in $T'$, then $x$ is a leaf of $T$. If $x$ has a degree-$1$ neighbor in $T$, then the connected component of $x$ in $T$ would consist entirely of just those two vertices and as a result, $T$ would have no non-leaf nodes. This cannot be the case since $|T'| \ge 6$. Thus, $x$ has no degree-$1$ neighbors, and therefore there are no bundles at $x$. As a result, the score of $x$ is $0$. Thus, the score of $T$ is the sum over all vertices $x$ in $T'$ of the score of $x$.

For $x \in T'$, define $d(x)$ to be the degree of $x$ in tree $T'$. For any vertex $x \in T'$, we can lower-bound the score of $x$ using casework:
\begin{itemize}
\item $d(x) = 0$. This would imply that $x$ is the only vertex in $T'$, directly contradicting the fact that $|T'| \ge 6$. Thus, we can conclude that this case is impossible.
\item $d(x) = 1$. The vertex $x$ in $T'$ has exactly one edge in $T$ leading to a non-leaf neighbor. This necessarily implies that all the other edges incident on $x$ form a bundle. Notice that $x$ has degree at least $5$ in $T$ since it is a vertex that was originally in $G$. Thus, the one bundle at $x$ has size at least $4$. The score of this bundle is then at least $3$, and so the score of $x$ is also at least $3 = 4 - d(x)$.
\item $d(x) = 2$. The vertex $x$ has exactly two edges in $T$ leading to non-leaf neighbors. These two edges separate all of the other edges incident on $x$ into at most two bundles. Notice that $x$ has degree at least $5$ in $T$ since it is a vertex that was originally in $G$. Thus there are at least $3$ edges in the (at most) two bundles at $x$. If there is one bundle, then the bundle has size at least $3$ and score at least $2$. If there are two bundles, then the total size of the two bundles is at least $3$, implying that the minimum possible total score of the two bundles is $0$ (which occurs in the case that one bundle has size $1$ and the other has size $2$). In all cases, the score of $x$ is at least $0 = 2-d(x)$.
\item $d(x) > 2$. The vertex $x$ has exactly $d(x)$ edges leading to non-leaf neighbors. These $d(x)$ edges separate all of the other edges incident on $x$ into at most $d(x)$ bundles. Each bundle has a score of at least $-1$, so $x$ has a score of at least $-d(x)$.
\end{itemize}

If we use $\mathbf{1}_{X}$ to represent the indicator function (which outputs $1$ if $X$ is true and $0$ otherwise), then the above results can be summarized as follows: the score of $x \in T'$ is bounded below by $(-d(x))\times\mathbf{1}_{d(x) > 2}+(4-d(x)) \times \mathbf{1}_{d(x) = 1} + (2-d(x)) \times \mathbf{1}_{d(x) = 2}$. Equivalently, we have that the score of $x \in T'$ is bounded below by $-d(x)+4 \times \mathbf{1}_{d(x) = 1} + 2 \times \mathbf{1}_{d(x) = 2}$. 

Adding this up, we see that the score of $T$ is at least
$$\sum_{x \in T'}\left(-d(x)+4\times \mathbf{1}_{d(x) = 1} + 2\times \mathbf{1}_{d(x) = 2}\right) = -\sum_{x \in T'}d(x)+4 \sum_{x \in T'}\mathbf{1}_{d(x) = 1} + 2\sum_{x \in T'}\mathbf{1}_{d(x) = 2}.$$
Since $\sum_{x \in T'}d(x)$ is the total degree of vertices in tree $T'$, this value is twice the total number of edges in $T'$. The terms $\sum_{x \in T'}\mathbf{1}_{d(x) = 1}$ and $\sum_{x \in T'}\mathbf{1}_{d(x) = 2}$ are the number of leaves and number of degree-$2$ vertices in $T'$.

Thus the score of $T$ is at least 
$$-2(\text{number of edges in}~T') + 4(\text{number of leaves of}~T') + 2(\text{number of degree-$2$ vertices in}~T').$$
As argued above, this implies our desired result: that the score of $T$ is positive.
\end{proof}

Next, we proceed to the other side of the contradiction: showing that the score of $T$ is non-positive.

To begin, we define an ordering $s_1, \ldots, s_{|S|}$ of the vertices in $S$ as follows:

\begin{definition}
\label{def:break_order}
Let $G_0 = G$. Then for $i = 1, \ldots, |S|$, define $s_i$ and $G_i$ as follows: Let $s_i$ be any vertex of $S$ that is on the boundary of the external face of $G_{i-1}$ and let $G_i$ be $G_{i-1}$ with vertex $s_i$ broken.
\end{definition}

\begin{lemma}
Definition~\ref{def:break_order} is well defined.
\end{lemma}

\begin{proof}
Notice that once a vertex is broken, it is no longer in the graph. Thus, it is impossible for the procedure given in Definition~\ref{def:break_order} to assign some element of $S$ to be both $s_i$ and $s_j$ for $i \ne j$.

With that said, in order to conclude that Definition~\ref{def:break_order} is well defined, it is sufficient to show that at each step, a choice of $s_i$ satisfying the conditions given in the definition is possible.

Fix $i \in \{1, \ldots, |S|\}$. Notice that $G_{i-1}$ is not a tree since that would mean that breaking a proper subset $\{s_1, \ldots, s_{i-1}\}$ of $S$ in $G$ yields a tree (in which case breaking the rest of $S$ would disconnect $G$). From this, we can conclude that $G_{i-1}$ has both an external face and at least one internal face. 

The points not in the external face form some set of connected regions; furthermore, this set is not empty since there is at least one internal face. Let $R$ be any such connected region. The boundary of $R$ must consist of a cycle of edges separating the external face from the internal faces inside $R$. This cycle of edges must contain at least one vertex of $S$ since otherwise the cycle would remain in $T$ after breaking every vertex of $S$ in $G$. Furthermore, since this cycle seperates the external face from an internal face, the cycle must be part of the boundary of the external face. 

Thus, when Definition~\ref{def:break_order} says to choose $s_i$ to be any vertex of $S$ that is on the boundary of the external face of $G_{i-1}$, this is well defined.
\end{proof}

\begin{definition}
We say that edge $e$ is an \emph{external edge} in $G_i$ if $e$ is an edge of $G_i$ with the external face on both sides. We say that $e$ is a \emph{boundary edge} of $G_i$ if the external face is on exactly one side of $e$. Finally, we say that $e$ is an \emph{internal edge} of $G_i$ if the external face is on neither side of $e$.  
\end{definition}

\begin{lemma}
Consider any edge $e$ incident on $s_i$ in graph $G_{i-1}$ and let $x$ be the other endpoint. When converting $G_{i-1}$ into $G_i$ by breaking $s_i$, 
\begin{itemize}
\item if $e$ is a boundary edge, then $e$ either joins one previously existing bundle at $x$ or becomes a new bundle at $x$ of size $1$.
\item if $e$ is an internal edge, then $e$ becomes a new bundle at $x$ of size $1$.
\end{itemize}
\end{lemma}

\begin{proof}
We begin with a proof by induction that every degree-$1$ vertex in $G_i$ is inside the external face for $i \in \{0, \ldots, |S|\}$. Certainly, it is the case that every degree-$1$ vertex in $G = G_0$ is inside the external face since $G$ has no degree-$1$ vertices. Then suppose all the degree-$1$ vertices of $G_{i-1}$ are inside the external face of $G_{i-1}$ for some $i$. We obtain $G_i$ from $G_{i-1}$ by breaking vertex $s_i$ on the boundary of the external face of $G_{i-1}$. As a result, the external face in $G_i$ is equal to the union of the faces touching $s_i$ in $G_{i-1}$ (including the external face). Thus, every degree-$1$ vertex that was inside the external face of $G_{i-1}$ is still inside the external face. Furthermore, every new degree-$1$ vertex (formed by breaking $s_i$) is also created inside the external face of $G_i$. By induction, we have our desired result: for $i \in \{0, \ldots, |S|\}$, every degree-$1$ vertex in $G_i$ is inside the external face.

Notice that $x$ must have at least one neighbor other than $s_i$ that has degree more than $1$ in $G_{i-1}$ because otherwise breaking $s_i$ would disconnect $x$ and its neighbors from the rest of the graph. This means that if we start at edge $e$ and go clockwise around $x$, we will eventually encounter some edge leading to a neighbor with degree more than $1$. Call this edge $e_+$. Similarly, starting at edge $e$ and going counterclockwise around $x$, we can let the first encountered edge leading to a neighbor with degree more than $1$ be $e_-$. There may or may not be a bundle at $x$ between $e$ and $e_+$. Similarly there may or may not be a bundle at $x$ between $e$ and $e_-$. In any case, when $s_i$ is broken, these bundles, if they exist, are merged, and edge $e$ is added to the one resulting bundle. 

Suppose that there is a bundle between $e$ and $e_+$. This means that there are degree-$1$ vertices in the face that is clockwise from $e$ around $x$. Then since degree-$1$ vertices are always found inside the external face, we can conclude that the face on that side of $e$ is the external face of $G_{i-1}$. Similarly, we can use the same logic to show that if there is a bundle between $e$ and $e_-$, then the region on the other side of $e$ is the external face of $G_{i-1}$. These results are shown in Figure~\ref{figure:bundle_possibilities}.

\begin{figure}[!htb]
    \centering
    \includegraphics[scale=.5]{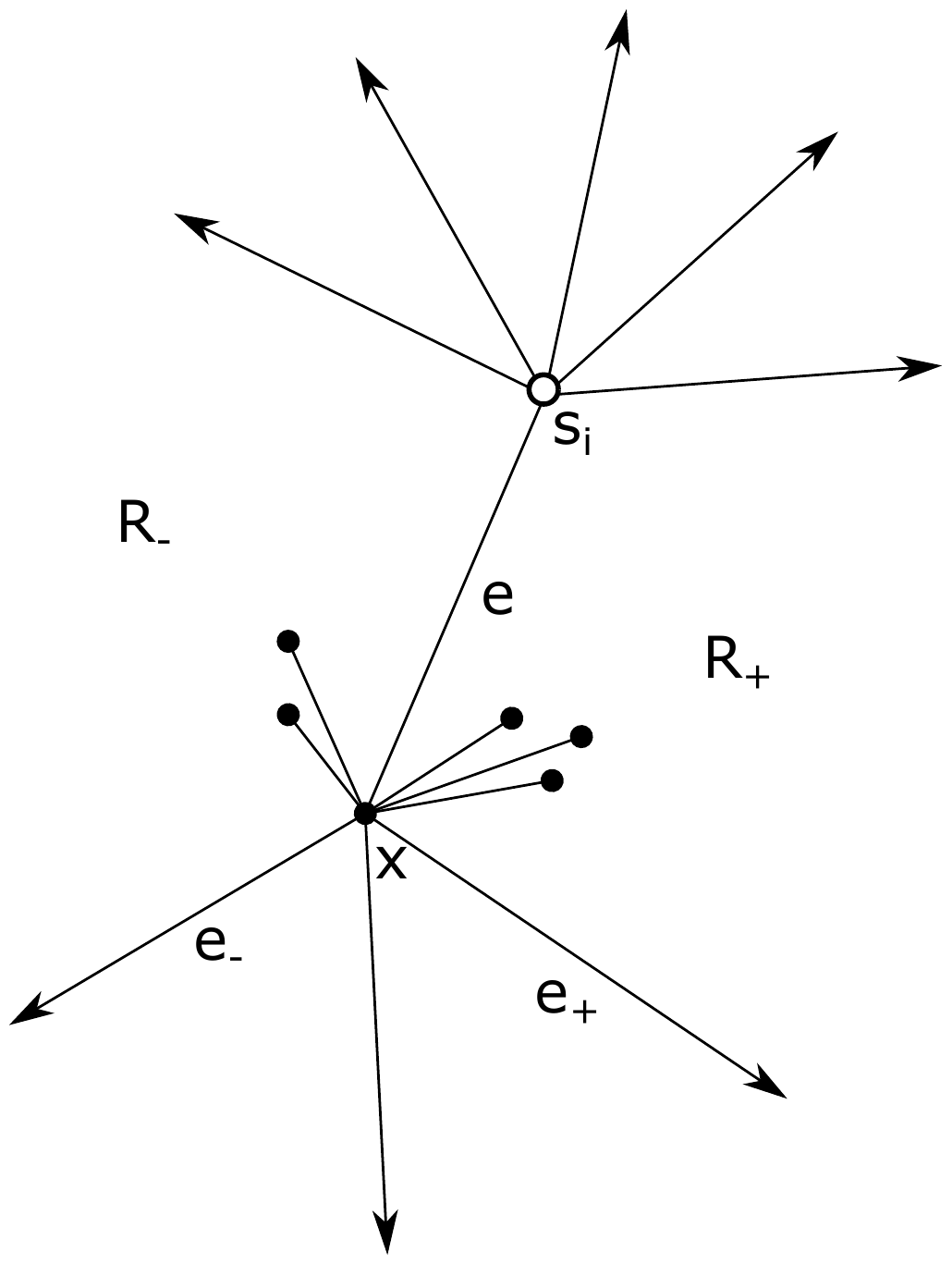}
    \caption{This figure shows vertices $s_i$ and $x$ with edge $e$ between them and edges $e_-$ and $e_+$ at $x$ as defined in the proof. Regions $R_+$ and $R_-$ are the regions on the two sides of $e$. The bundle shown in region $R_+$ at $x$ can be present only if $R_+$ is the external face. Similarly, the bundle shown in region $R_-$ at $x$ can be present only if $R_-$ is the external face.}
    \label{figure:bundle_possibilities}
\end{figure}

Next, consider the two cases in which edge $e$ is a boundary or internal edge. If $e$ is an internal edge, then it has the external face on zero sides, and so neither the bundle at $x$ between $e$ and $e_+$ nor the bundle at $x$ between $e$ and $e_-$ can exist. Then when $s_i$ is broken, edge $e$ forms a new bundle of size $1$ at $x$. If $e$ is a boundary edge, then it has the external face on exactly one side. In this case, at most one of the two bundles at $x$ (between $e$ and $e_+$ or $e$ and $e_-$) can exist. Then when $s_i$ is broken, edge $e$ either forms a new bundle of size $1$ (if neither bundle existed) or is added to a previously existing bundle at $x$ (if one of the two bundles existed). This is exactly what we wished to show.
\end{proof}

\begin{lemma}
When going from $G_{i-1}$ to $G_i$, the score cannot increase.
\end{lemma}

\begin{proof}
The only vertices in $G_i$ that are not in $G_{i-1}$ are the degree-$1$ vertices which replace $s_i$ when it is broken. The only vertex in $G_{i-1}$ that is not in $G_i$ is $s_i$. We claim that the neighbors of all of these vertices have degree not equal to $1$, and therefore that each of these vertices has no bundles and thus a score of $0$. Suppose for the sake of contradiction that some one of these vertices which are present in one graph but not the other has a neighbor of degree $1$. This means that either a degree-$1$ vertex in $G_i$ or $s_i$ in $G_{i-1}$ has a degree-$1$ neighbor. If $s_i$ has a degree-$1$ neighbor, then when it is broken, one of the degree-$1$ vertices replacing it inherits that neighbor. Thus in all cases, some two degree-$1$ vertices in $G_i$ are neighbors. These two vertices form a connected component, implying that $G_i$ is not connected. This contradicts the fact that $S$ is a solution to the TRVB instance, so as desired, none of the vertices in question have degree-$1$ neighbors and so all of these vertices have score $0$.

Thus, since every vertex in exactly one of $G_{i-1}$ and $G_i$ has a score of $0$, the difference in score between $G_{i-1}$ and $G_i$ is equal to the cumulative difference in score over all vertices that are in both of these graphs. Suppose $x$ is a vertex of both graphs that does not neighbor $s_i$. Breaking $s_i$ does not affect $x$ or the degrees of the neighbors of $x$. Thus the bundles at $x$ remain the same in $G_{i-1}$ and $G_{i}$. In other words, there is no change in score at vertex $x$.

Then to compute the difference in scores between $G_{i-1}$ and $G_i$ we must simply compute the cumulative difference in scores between $G_{i-1}$ and $G_i$ of vertices $x$ that neighbor $s_i$. If the edge between $s_i$ and $x$ is an internal edge, then by the previous lemma, the change in bundles at $x$ between $G_{i-1}$ and $G_i$ is that a new bundle of size $1$ is added. This decreases the score of $x$ by $1$. If the edge between $s_i$ and $x$ is a boundary edge, then by the previous lemma, the change in bundles at $x$ between $G_{i-1}$ and $G_i$ is either that a new bundle of size $1$ is added or that the size of some one bundle is increased by $1$. Thus, the score of $x$ either decreases by $1$ (if a new bundle is added), increases by $2$ (if a bundle of size $1$ becomes a bundle of size $2$), or increases by $1$ (if a bundle of size at least $2$ increases in size by $1$). Below, we will show that exactly two of the edges incident on $s_i$ in $G_{i-1}$ are boundary edges and that the rest are internal edges. Since the degree of $s_i$ is at least $6$, this means that exactly two neighbors of $x$ will have their score increase by at most $2$ and that all the other neighbors (of which there are at least $4$) will have their score decrease by $1$. In other words, when going from $G_{i-1}$ to $G_i$, the score cannot increase, which is exactly the statement of this lemma.

All that's left is to show that exactly two of the edges incident on $s_i$ in $G_{i-1}$ are boundary edges and that the rest are internal edges. 

Consider the faces which touch at $s_i$ in $G_{i-1}$. Since $s_i$ was chosen to be on the external face, one of these faces is the external face. Let this face be $F_0$, and let the other faces clockwise around $s_i$ be $F_1, F_2, \ldots F_{d-1}$ where $d$ is the degree of $s_i$ in $G_{i-1}$. Let $e_j$ be the edge separating face $F_j$ from the previous face clockwise around $s_i$. See Figure~\ref{figure:faces_example} for an example. Finally, let $x_j$ be the endpoint of $e_j$ other than $s_i$.

\begin{figure}[!htb]
    \centering
    \includegraphics[scale=.5]{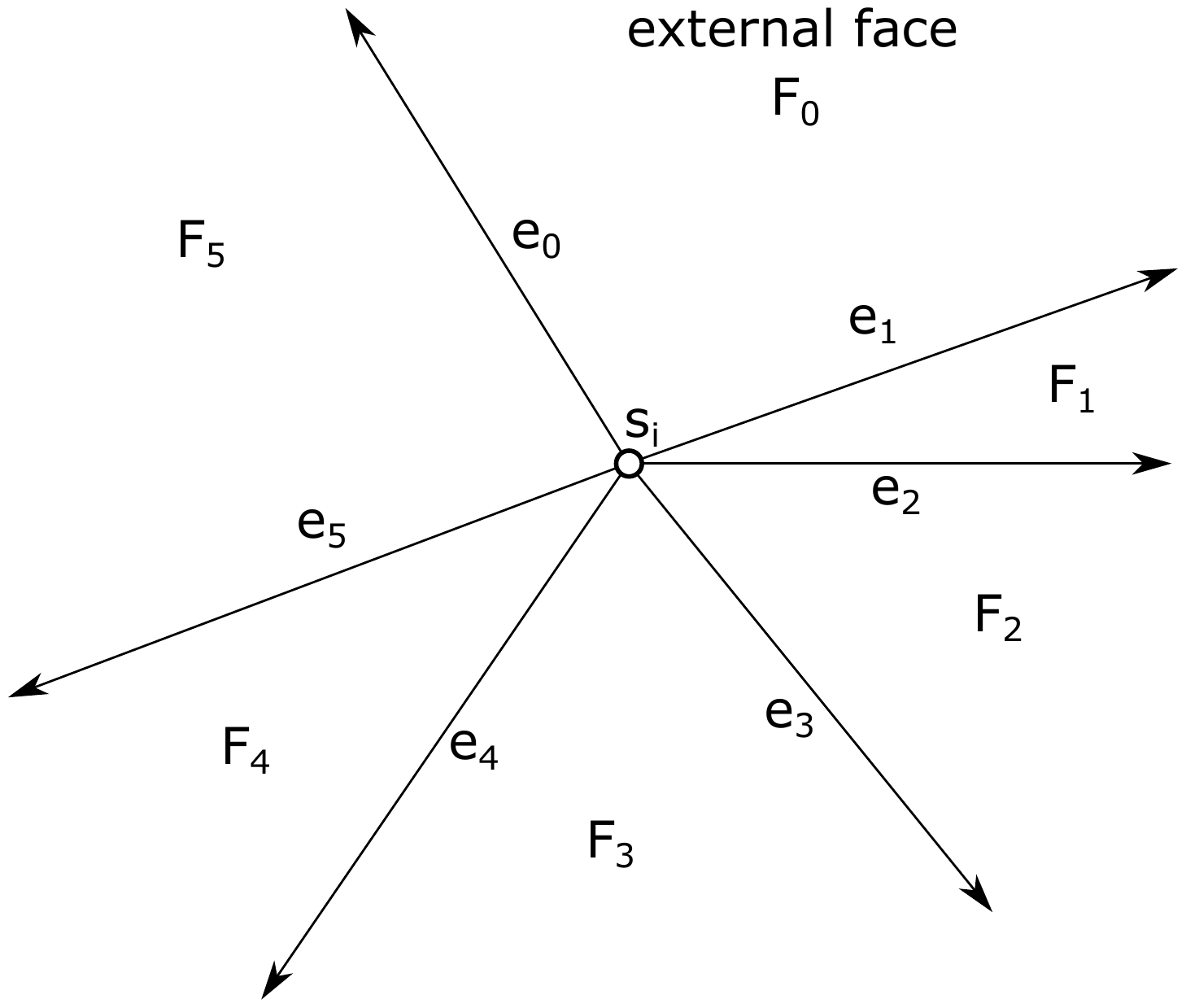}
    \caption{This figure shows an example vertex $s_i$ of $G_{i-1}$ together with the edges and faces meeting at $s_i$.}
    \label{figure:faces_example}
\end{figure}

Since $G_i$ is connected, there exists a path from $x_0$ to $x_1$ in $G_i$. This path will also exist in $G_{i-1}$ since a path cannot use any vertices of degree $1$ and every vertex in $G_i$ whose degree is not $1$ is also in $G_{i-1}$. There is another path in $G_{i-1}$ from $x_0$ to $x_1$: namely the path $x_0, s_i, x_1$. Together, these two paths form a cycle in $G_{i-1}$ including the two edges $e_0$ and $e_1$. Either the faces that are between $e_0$ and $e_1$ clockwise around $s_i$ (in particular face $F_0$) or counterclockwise around $s_i$ (all the other $F_j$s) must be on the interior of this cycle. Since $F_0$ is the external face, we know that it cannot be on the interior of a cycle. Thus each other $F_j$ must be on the interior of the cycle and therefore cannot equal the external face. Notice that edges $e_0$ and $e_1$ have the external face on exactly one side and that each other edge $e_j$ has the external face on neither side. As desired, we have shown that exactly two of the edges incident on $s_i$ in $G_{i-1}$ are boundary edges and that the rest are internal edges. 
\end{proof}

\begin{lemma}
The score of $T$ is not positive.
\end{lemma}

\begin{proof}
$G_0 = G$ has no degree-$1$ vertices and therefore has no bundles. Thus, the score of $G_0$ is $0$. By the previous lemma, the score of $G_i$ is non-increasing as a function of $i$. We can immediately conclude that the score of $G_{|S|}$, the graph formed by breaking vertices $s_1, \ldots, s_{|S|}$ is not positive. But $s_1, \ldots, s_{|S|}$ are all the vertices in $S$, so $G_{|S|} = T$ and as desired, the score of $T$ is not positive.
\end{proof}

Notice that we have seen two directly contradictory lemmas: we have shown that the score of $T$ is both positive and not positive. By contradiction, we can conclude that $S$, the solution to Planar Graph $(\{6,7,8,\ldots\}, \{5,6,7,\ldots\})$-TRVB instance $G$, cannot exist. Thus, we have that

\begin{lemma}
Planar Graph $(\{6,7,8,\ldots\}, \{5,6,7,\ldots\})$-TRVB is polynomial-time solvable.
\end{lemma}

\begin{proof}
We have shown that for any instance of Planar Graph $(\{6,7,8,\ldots\}, \{5,6,7,\ldots\})$-TRVB, the correct answer is ``no.'' Thus rejecting all inputs (a polynomial-time algorithm) solves Planar Graph $(\{6,7,8,\ldots\}, \{5,6,7,\ldots\})$-TRVB.
\end{proof}

From this, we obtain our desired result.

\begin{theorem}
If $b > 5$ for every $b \in B$ and $u > 4$ for every $u \in U$, then Planar Graph $(B,U)$-TRVB can be solved in polynomial time.
\end{theorem}

\begin{proof}
This follows immediately from the previous lemma together with Lemma~\ref{lemma:trivial_reductions}.
\end{proof}

\section{TRVB and the Hypergraph Spanning Tree problem}
\label{section:hypergraph}

The overall purpose of this section is to demonstrate the connection between the TRVB problem and the Hypergraph Spanning Tree problem. 

\begin{definition}
A \emph{hypergraph} is a pair $(V, E)$ where every element $e$ of $E$ is a subset of $V$. The elements of $V$ are called \emph{vertices} and the elements of $E$ are called \emph{edges} or \emph{hyperedges}. A vertex $v$ is an \emph{endpoint} of a hyperedge $e$ whenever $v \in e$. The \emph{incidence graph} associated with a hypergraph $(V, E)$ is a bipartite graph whose two parts are $V$ and $E$. In the incidence graph, there is an edge between $e \in E$ and $v \in V$ if and only if $v$ is an endpoint of $e$. A set $S$ of hyperedges in a hypergraph $(V, E)$ is a \emph{hypergraph spanning tree} if the subgraph of the incidence graph induced by vertices $V \cup S$ is a tree (or in other words if removing all other edge nodes from the incidence graph yields a tree). A hypergraph is $r$-regular if every vertex is an endpoint of exactly $r$ edges, and $u$-uniform if every edge has exactly $u$ endpoints.
\end{definition}

\begin{problem}
The \emph{Hypergraph Spanning Tree} problem asks given a hypergraph whether there exists a hypergraph spanning tree in that hypergraph.
\end{problem}

We show how to reduce from the TRVB problem to the Hypergraph Spanning Tree problem. Suppose we are given an instance of TRVB consisting of multigraph $M$ (with vertices labeled as breakable or unbreakable). Let $M'$ be the multigraph produced by repeatedly identifying an edge between two different unbreakable vertices and contracting that edge until no more such edges remain. Clearly, contracting an edge between two unbreakable vertices cannot change the answer to the TRVB problem, so $M'$ has the same answer to TRVB as $M$. Next, we construct $M''$ from $M'$ by inserting a degree-$2$ unbreakable vertex into every edge whose endpoints are both breakable. Again, this is an operation which does not affect the answer to TRVB, so $M''$ has the same answer to TRVB as $M$.

There are two cases: either $M''$ has a self-loop at an unbreakable vertex, or $M''$ is bipartite with breakable and unbreakable vertices as the two parts. In the first case, we can trivially deduce that the answer to TRVB instance $M''$ (and therefore also to TRVB instance $M$) is ``no''. Thus, in that case, we simply output any ``no'' instance of the Hypergraph Spanning Tree problem. In the second case, we construct a hypergraph $H$ whose incidence graph is $M''$ with edges of $H$ represented by breakable vertices and vertices of $H$ represented by unbreakable vertices. Outputting this hypergraph $H$ concludes the reduction. 

\begin{lemma}
  This reduction
  is correct:
  the input $M$ is a ``yes'' instance to TRVB
  if and only if
  the output $H$ has a hypergraph spanning tree.
\end{lemma}

\begin{proof}
We wish to show that $H$ has a hypergraph spanning tree if and only if the equivalent instance $M''$ is a ``yes'' instance of TRVB. By definition, $H$ has a hypergraph spanning tree if and only if there exists a set of edges such that removing the vertices corresponding to those edges from the incidence graph of $H$ yields a tree. But the incidence graph of $H$ is $M''$, and the vertices in $M''$ corresponding to edges in $H$ are exactly the breakable vertices. Thus, there exists a hypergraph spanning tree in $H$ if and only if there exists a choice of breakable vertices in $M''$ such that removing those vertices yields a tree. As a result, it is sufficient to show that $M$ is a ``yes'' instance to TRVB if and only if there exists a choice of breakable vertices in $M''$ such that removing those vertices yields a tree. 

Suppose that there exists a set of breakable vertices $S$ in $M''$ such that removing those vertices yields a tree. For every vertex $v$ in $S$, we can add a degree-$1$ neighbor in the resulting tree to every neighbor of $v$ (since the neighbors of $v$ are all unbreakable and therefore occur in this tree). Clearly, adding degree-$1$ neighbors does not change the fact that the result is a tree. Note however, that the resulting tree is exactly the graph that we would obtain if we broke every vertex of $S$ in $M''$: in obtaining this tree, every vertex $v$ in $S$ was removed and then replaced by degree-$1$ neighbors for the neighbors of $v$ (or in other words every vertex $v$ in $S$ was broken). Thus we have shown that if there exists a choice of breakable vertices in $M''$ such that removing those vertices yields a tree then breaking the same vertices also yields a tree.

Next, suppose that there exists a set of breakable vertices $S$ in $M''$ such that breaking those vertices yields a tree. If we remove a set of degree-$1$ vertices from a tree, then we always obtain another tree. Remove the degree-$1$ vertices that were added while breaking the vertices of $S$ from the resulting tree. The result (another tree) is exactly the graph that would be obtained if we were to simply remove the vertices of $S$ from $M''$ instead of breaking them. Thus we have shown that if there exists a choice of breakable vertices in $M''$ such that breaking those vertices yields a tree then removing the same vertices also yields a tree.

We have shown that there exists a choice of breakable vertices in $M''$ such that removing those vertices yields a tree if and only if there exists a choice of breakable vertices in $M''$ such that breaking those vertices yields a tree. Thus, we can conclude that $H$ has a hypergraph spanning tree if and only if $M''$ is a ``yes'' instance of TRVB. From this, we see that the reduction is correct.
\end{proof}

Consider what happens if we apply this reduction to an instance of TRVB whose breakable vertices all have degree at most $3$. In this case, the output instance of the Hypergraph Spanning Tree problem is a hypegraph $H$ whose edges each have at most $3$ endpoints. In other words, the above reduction can also be seen as a reduction from any problem $(B, U)$-TRVB with $B \subseteq \{1,2,3\}$ to a version of the Hypergraph Spanning Tree problem in which the hypergraphs are restricted to have only edges with at most $3$ endpoints. The Hypergraph Spanning Tree problem in such hypergraphs is known to be polynomial-time solvable (see \cite{lovasz}), so we can immediately conclude the following:

\begin{corollary}
$(B, U)$-TRVB with $B \subseteq \{1,2,3\}$ is polynomial-time solvable.
\end{corollary}

Similarly, consider what happens if we apply this reduction to an instance of planar TRVB with only breakable vertices whose vertices all have degree $k$ for some fixed $k$. In this case, the output instance of the Hypergraph Spanning Tree problem is a $k$-uniform $2$-regular hypegraph $H$ whose incidence graph is planar. In other words, the above reduction can also be seen as a reduction from Planar $(\{k\}, \emptyset)$-TRVB to a version of the Hypergraph Spanning Tree problem in which the hypergraphs are restricted to be $k$-uniform and $2$-regular and to have planar incidence graphs. Applying the fact that Planar $(\{k\}, \emptyset)$-TRVB is NP-hard for any $k \ge 4$, we immediately obtain that 

\begin{corollary}
The Hypergraph Spanning Tree problem is NP-complete in $k$-uniform $2$-regular hypergraphs for any $k \ge 4$, even when the incidence graph of the hypergraph is planar.
\end{corollary}

This improves the previously known result that the Hypergraph Spanning Tree problem is NP-complete in $k$-uniform hypergraphs for any $k \ge 4$ (see \cite{promel2002steiner}).

\section*{Acknowledgments}
We would like to thank Zachary Abel and Jayson Lynch for their helpful discussion about this research. We would also like to thank Yahya Badran for pointing out the connection between TRVB and the Hypergraph Spanning Tree problem.

\bibliographystyle{plain}
\bibliography{vertex_breaking}

\end{document}